%% file: main.tex
\documentclass{article}

\usepackage[natbib,backend=bibtex]{biblatex}
\bibliography{LRU_cache_analysis_POPL2019}

\usepackage{algorithm}
\usepackage{algpseudocode}
\usepackage{paralist}

\usepackage{bm}

\usepackage{tikz}
\usetikzlibrary{shapes.geometric}

\usepackage{authblk}

\usepackage{booktabs}   
\usepackage{hyperref}
\usepackage{subcaption} 

\usepackage{moreverb}

\usepackage{amssymb,amsmath,amsthm}
\theoremstyle{plain}
\newtheorem{theorem}{Theorem}[section]

\newtheorem{lemma}[theorem]{Lemma}
\newtheorem{corollary}[theorem]{Corollary}
\theoremstyle{remark}
\newtheorem{example}[theorem]{Example}
\theoremstyle{definition}
\newtheorem{definition}[theorem]{Definition}

\input{macros}

\title{Fast and exact analysis for LRU caches%
  \thanks{This work was partially supported by the \href{http://erc.europa.eu/}{European Research Council} under the European Union's Seventh Framework Programme (FP/2007-2013) / ERC Grant Agreement nr. 306595 \href{http://stator.imag.fr}{``STATOR''}, and by the Deutsche Forschungsgemeinschaft as part of the project PEP: Precise and Efficient Prediction of Good Worst-case Performance for Contemporary and Future Architectures.}
  \thanks{This is an extended version of work presented at POPL 2019: it includes appendices not present in the ACM publication.}}

\author{Valentin Touzeau}
\author{Claire Ma\"{\i}za}
\author{David Monniaux}
\affil{Univ. Grenoble Alpes, CNRS, Grenoble INP\footnote{Institute of Engineering Univ. Grenoble Alpes}, VERIMAG, 38000 Grenoble, France}
\author{Jan Reineke}
\affil{Saarland University, Saarbr\"ucken, Germany}
  
\newcommand{\semibold}{\fontseries{sb}\selectfont}

\begin{document}
\maketitle

\begin{abstract}
\input{abstract}
\end{abstract}

\input{introduction}
\input{problem}
\input{motivating_example}
\input{fixpoint}
\input{zdd}
\input{extensions}
\input{complexity}
\input{implementation}
\input{relatedwork}
\input{conclusion}

\printbibliography

\appendix
\clearpage
\input{sharing}
\clearpage
\input{exact_analyses}

\end{document}

%% file: macros.tex
\newcommand{\oldnewFigure}[3]{#3} 
\newcommand{\oldnew}[2]{#2}	

\newcommand{\vneg}[1]{\bar{#1}}
\newcommand{\cacheWays}{N}
\newcommand{\noAccess}{\varepsilon}
\newcommand{\noLine}{\varepsilon}
\newcommand{\absent}{\mathcal{A}}
\newcommand{\locations}{A}
\newcommand{\software}[1]{\textsc{#1}}

\newcommand{\todo}[1]{{\color{red}(TODO: #1)}}

\usepackage{catoptions}
\makeatletter

\def\Autoref#1{%
  \begingroup
  \edef\reserved@a{\cpttrimspaces{#1}}%
  \ifcsndefTF{r@#1}{%
    \xaftercsname{\expandafter\testreftype\@fourthoffive}
      {r@\reserved@a}.\\{#1}%
  }{%
    \ref{#1}%
  }%
  \endgroup
}
\def\testreftype#1.#2\\#3{%
  \ifcsndefTF{#1autorefname}{%
    \def\reserved@a##1##2\@nil{%
      \uppercase{\def\ref@name{##1}}%
      \csn@edef{#1autorefname}{\ref@name##2}%
      \autoref{#3}%
    }%
    \reserved@a#1\@nil
  }{%
    \autoref{#3}%
  }%
}
\makeatother

%% file: abstract.tex
For applications in worst-case execution time analysis and in security, it is desirable to statically classify memory accesses into those that result in cache hits, and those that result in cache misses.
Among cache replacement policies, the least recently used (LRU) policy has been studied the most and is considered to be the most predictable. 

The state-of-the-art in LRU cache analysis presents a tradeoff between precision and analysis efficiency:
The classical approach to analyzing programs running on LRU caches, an abstract interpretation based on a range abstraction, is very fast but can be imprecise.
An exact analysis was recently presented, but, as a last resort, it calls a model checker, which is expensive.

In this paper, we develop an analysis based on abstract interpretation that comes close to the efficiency of the classical approach, while achieving exact classification of all memory accesses as the model-checking approach.
Compared with the model-checking approach we observe speedups of several orders of magnitude.
As a secondary contribution we show that LRU cache analysis problems are in general NP-complete.



%% file: introduction.tex
\section{Introduction}
\subsection{Motivation}
Due to technological developments, the latency of an access to DRAM-based main memory has long been much higher than the latency of an individual computation on the CPU.
The most common solution to bridge this ``memory gap'' is to include a hierarchy of cache memories between the CPU and main memory, meant to speed up accesses to frequently required code and operands.

In the presence of caches, the latency of an individual memory access may vary considerably depending on whether the access is a \emph{cache hit}, i.e., it can be served from an on-chip cache memory, or a \emph{cache miss}\footnote{An often quoted figure is that a cache miss is $100$ times slower than a cache hit.
  We have a simple program whose memory access pattern can be changed by a numeric parameter while keeping exactly the same computations;
  we timed it with a cache-favorable sequential access pattern compared to a cache-unfavorable one.
  Depending on the machine and processor clocking configuration, the ratio of the two timings varies between 13 and 40 if one core is used, up to 50 with two cores.}, i.e., it has to be fetched from the next level cache or DRAM-based main memory.

The purpose of \emph{cache analysis} is to statically classify every memory access at every machine-code instruction in a program into one of the following three classes:
\begin{enumerate}
	\item ``always hit'': each dynamic instance of the memory access results in a cache hit;
	\item ``always miss'': each dynamic instance of the memory access results in a cache miss;
	\item there exist dynamic instances that result in a cache hit and others that result in a cache miss.
\end{enumerate}
This is of course, in general, an undecidable question
; so all analyses involve some form of abstraction, and may classify some accesses as ``unknown''.
An analysis is deemed more precise than another if it produces fewer unknowns.

For the certification of safety-critical real-time applications, it is often necessary to bound a program's \emph{worst-case execution time} (WCET).
For instance, if a control loop runs at 100 Hz, then it is imperative to show that the program's WCET is less than $0.01$ seconds.
In architectures involving caches, i.e., any modern architecture except the lowest-performance ones, such WCET analysis \cite{DBLP:journals/tecs/WilhelmEEHTWBFHMMPPSS08}%
\footnote{WCET static analysis tools include, among others, \href{https://www.absint.com/ait/}{\software{aiT}}, an industrial tool from Absint GmbH, and \href{http://otawa.fr/}{\software{Otawa}}, an academic tool.}
must thus take into account caches.
For pipelined and superscalar architectures, it is very important to have precise information about the cache behavior, since pipeline analysis must consider the two cases ``cache hit'' and ``cache miss'' for any memory access that cannot be shown to ``always hit'' or ``always miss''~\cite{Lundqvist99, Reineke06}, leading to a state explosion.
Thus, imprecise cache analysis may have two adverse effects on WCET analysis:
\begin{inparaenum}[(a)]
\item excessive overestimation of the WCET compared to the true WCET%
\footnote{An industrial user may suspect this when the upper bound on WCET given by the tools is far from experimental timings.
This may discourage the user from using static analysis tools.}
\item excessively high analysis time due to state explosion.
\end{inparaenum}
Improvements to cache analysis precision are thus of high importance in this respect; but they must come at reasonable cost.

Caches also give rise to side channels that can be used to extract or transmit sensitive information.
For example, cache timing attacks on software implementations of the Advanced Encryption Standard \cite{Bernstein_2005} were one motivation for adding specific hardware support for that cipher to the x86 instruction set~\cite{Mowery:2012:AXC:2381913.2381917}.
Cache analysis may help identify possibilities for such \emph{side-channel attacks} and quantify the amount of information leakage~\cite{Doychev2015,Doychev:2017:RAS:3062341.3062388};
improved precision in cache analysis translates into fewer false alarms and tighter leakage bounds.


\subsection{Cache Organization, Cache Analysis, and the State-of-the-Art}
Instruction and data caches are usually set-associative and thus partitioned into \emph{cache sets}.
Each memory block may reside in exactly one of these cache sets, determined by a simple computation from its address.
Each cache set consists of multiple \emph{cache lines}, which may each be used to store a single memory block.
The number of cache lines $\cacheWays$ in each cache set is known as the \emph{number of ways} or the \emph{associativity} of the cache.
Upon a cache miss to a memory block that maps into a full cache set, one of the $\cacheWays$ cached memory blocks must be evicted from the set.
There exist several \emph{policies} for choosing which memory block to evict.
In this article, we focus on the \emph{least recently used} policy (LRU):
the least recently used memory block within a cache set is evicted.
LRU has been extensively studied in the literature and is frequently used in practice, e.g., in processors such as the \software{MPC603E}, the \software{TriCore17xx}, or the \software{TMS320C3x}.

As noted before, cache analysis is in general undecidable
; so all analyses involve some form of abstraction.
Most work on cache analysis separates the concerns of
\begin{enumerate}[(a)]
\item\label{enum:cf_analysis} control-flow analysis, determining which execution paths can be taken, 
\item\label{enum:pointer_analysis} pointer analysis, determining which memory locations may be accessed by instructions, and
\item\label{enum:cache_analysis}  cache analysis proper.
\end{enumerate}

Concerns (\ref{enum:cf_analysis}) and (\ref{enum:pointer_analysis}) are in general undecidable, so safe and terminating analyses use some form of over-approximation --- they may consider some paths to be feasible when actually they are not, or that pointers may point to memory locations when actually they cannot.
In our case, as in many other works, we assume we are given a control-flow graph $G$ of the program decorated with the memory locations that are possibly accessed by the program, but without the functional semantics (arithmetic, guards etc.);
the executions of this control-flow graph are thus a superset of those of the program.
This is what we will consider as a concrete semantics --- though we shall sketch, as future work, in \Autoref{ref:future_work}, how to recover some of the precision lost by using that semantics by reintroducing information about infeasible paths.

In this paper, we focus on (\ref{enum:cache_analysis}) cache analysis proper for LRU caches.
When it comes to LRU cache analysis, the state-of-the-art currently presents a tradeoff between \emph{precision} and \emph{analysis efficiency}:

The classical approach to the static analysis of LRU caches \cite{DBLP:journals/rts/FerdinandW99} is a {\semibold highly-efficient abstract interpretation} that essentially keeps for each block a range of how many other blocks have been used more recently.
This analysis is exact for straight-line programs, but loses precision in general when tests are involved: the join operation adds spurious states, which may translate into classification of memory accesses as ``unknown,'' whereas, with respect to the concrete semantics of~$G$, they should be classified as ``always hit'' or ``always miss''.


Recently, \citet{DBLP:conf/cav/TouzeauMMR17} proposed an exact analysis, i.e., it exactly classifies memory accesses with respect to the concrete semantics, into ``always hit'', ``always miss'', or ``hits or misses depending on the execution''.
Their approach encodes the concrete cache state transitions into a reachability problem, fed to a symbolic model checker.
Since that approach was {\semibold slow}, they also proposed a fast abstract pre-analysis able to detect cases where an access ``hits or misses depending on the execution''.
The model-checking algorithm is then only applied to the relatively infrequent cases where accesses are still classified as ``unknown'' by their new analysis and the classical ones by Ferdinand and Wilhelm.



\subsection{Contributions}

In this paper, we develop an analysis based on abstract interpretation that comes close to the efficiency of the classical approach by \citet{DBLP:journals/rts/FerdinandW99} while achieving exact classification of all memory accesses as the model-checking approach by \citet{DBLP:conf/cav/TouzeauMMR17}.
In other terms, we introduce an exact and scalable analysis by carefully refining the abstraction and using suitable algorithms and data structures.

Our main contribution is the introduction of a new exact abstraction for LRU caches that is based on a partial order of cache states.
To classify cache misses (cache hits), it is sufficient to only keep minimal (maximal) elements w.r.t. this partial order.
As a consequence, the abstraction may be exponentially more succinct than the model-checking approach followed by \citet{DBLP:conf/cav/TouzeauMMR17}.

We improve the focused semantics of  \citet{DBLP:conf/cav/TouzeauMMR17} by removing subsumed elements with upward and downward closures. 
This form of convergence acceleration preserves the precision of the final classification. 

We discuss a suitable data structure for this abstraction based on zero-suppressed binary decision diagrams (ZDDs), and an implementation on top of \software{Otawa}~\cite{DBLP:conf/seus/BallabrigaCRS10} and \software{Cudd}~\cite{DBLP:journals/sttt/Somenzi01}.
Our experimental evaluation shows an analysis speedup of up to \oldnew{602}{950} compared with the prior exact approach by \citet{DBLP:conf/cav/TouzeauMMR17}.
The geometric mean of the speedup across all studied benchmarks is at least \oldnew{24}{9}\footnote{On a number of benchmarks the prior exact approach timed out at 12 hours. For these benchmarks, we conservatively assume an analysis time of 12 hours, and may thus underestimate the actual speedup, had the analysis been run to completion.}.
On the other hand, compared with the imprecise age-based analysis of \citet{DBLP:journals/rts/FerdinandW99} we observe an average slowdown across all benchmarks of only \oldnew{3.53}{3.46}.

Our secondary contribution is a proof that both of the problems that we address (existence of a trace leading to a cache hit, existence of a trace leading to a cache miss) are NP-complete.
To the best of our knowledge, this was not known previously, whereas it justifies using imprecise abstractions, as in the traditional age-based analyses,  and/or algorithms with non-polynomial worst-case complexity, as in our analysis.

\subsection{Outline}
In \Autoref{sec:problem} we define the static analysis problem for LRU caches.
In \Autoref{sec:motivating_example} we illustrate how the results of our analysis are more precise than those of classical analysis on a small example.
In \Autoref{sec:fixpoint} we reformulate this problem as a least fixpoint, then give two exact abstractions, one for ``always hit'', the other for ``always miss'' results, each of which can be implemented by computations over antichains.
In \Autoref{sec:zdd} we explain the algorithms and data structures used for the upward and downward closures, using and extending zero-suppressed binary decision diagrams (ZDDs).
In \Autoref{sec:extensions} we describe a few possible extensions and variants of our approach.
In \Autoref{sec:complexity} we discuss complexity issues and show that the analysis problems we solve are NP-hard, thus justifying the use of potentially exponentially large ZDDs.
In \Autoref{sec:implementation} we describe our implementation and our experimental results.
In \Autoref{sec:related_work} we discuss related work.
We conclude the paper sketching possible avenues for future work in \Autoref{sec:conclusion}.

%% file: problem.tex
\section{Problem Setting}
\label{sec:problem}

\begin{figure}
\begin{subfigure}{0.45\textwidth}
\begin{center}
    \begin{tikzpicture}[node distance=4em,->,auto]
      \node (start) [draw,diamond] { $\emptyset$ };
      \node (q0) [right of=start] { $v_0$ };
      \node (q1) [above right of=q0] { $v_1$ };
      \node (q3) [below right of=q1] { $v_3$ };
      \node (q2) [below right of=q0]  { $v_2$ };
      \path (start) edge node {$a$} (q0);
      \path (q0) edge node { $a$ } (q1);
      \path (q2) edge node { $c$ } (q0);
      \path (q1) edge node { $b$ } (q2);
      \path (q1) edge node { $d$} (q3);
      \path (q3) edge node { $e$} (q2);
    \end{tikzpicture}
\end{center}
\caption{Original control-flow graph for two cache sets:
$\{a,e\}$ and $\{b,c,d\}$.}

\label{fig:CFG}
\end{subfigure}
\hfill
\begin{subfigure}{0.45\textwidth}%
\begin{center}
    \begin{tikzpicture}[node distance=4em,->,auto]
      \node (start) [draw,diamond] { $\emptyset$ };
      \node (q0) [right of=start] { $v_0$ };
      \node (q1) [above right of=q0] { $v_1$ };
      \node (q3) [below right of=q1] { $v_3$ };
      \node (q2) [below right of=q0]  { $v_2$ };
      \path (start) edge node {$a$} (q0);
      \path (q0) edge node { $a$ } (q1);
      \path (q2) edge node { $\noAccess$ } (q0);
      \path (q1) edge node { $\noAccess$ } (q2);
      \path (q1) edge node { $\noAccess$} (q3);
      \path (q3) edge node { $e$} (q2);
    \end{tikzpicture}
\end{center}
\caption{The same control-flow graph focused for cache set $\{a,e\}$.}
\end{subfigure}

\caption{Slicing of a control-flow graph according to a cache set}
\label{fig:CFG-slicing}
\end{figure}
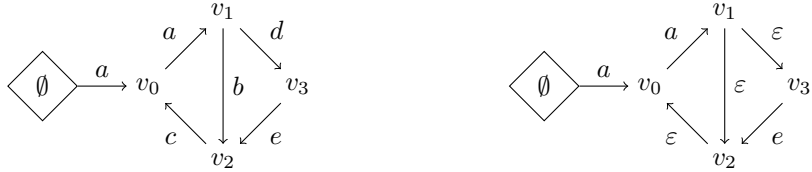

\begin{figure}
\begin{tikzpicture}[->,node distance=8em, auto]
\node (start) [draw,diamond] { $\emptyset$ };
\node (q1) [right of=start] { $(a,\noLine,\noLine,\noLine)$ };
\node (q2) [right of=q1, align=center] { $v_2$:\\ $(b,a,\noLine,\noLine)$\\$(c,a,\noLine,\noLine)$};
\node (q3) [right of=q2] { };
\path (start) edge node {$a$} (q1);
\path (q1) edge[bend left] node {$b$} (q2);
\path (q1) edge[bend right] node {$c$} (q2);
\path (q2) edge[dotted] (q3);
\end{tikzpicture}

\caption{At $v_2$, two cache states are possible, according to the path of arrival. In both of them $a$ is present, so $a$ is a ``must hit'' on edges going out of $v_2$. $b$ is present in one of them, so it is a ``may hit''. $d$ is present in neither, so it is a ``must miss''.}
\label{fig:must_may_hit}
\end{figure}
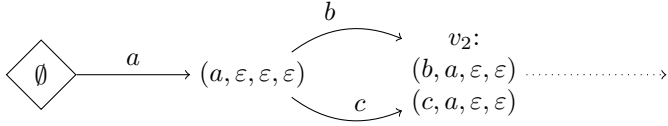

As is common in cache analysis, we assume the following analyses have already been performed:
\begin{enumerate}[(a)]
\item A control-flow graph has been reconstructed from the machine code of the program under analysis (perhaps with some knowledge about the compiler and/or dumps of its intermediate representation).
\item For instruction cache analysis, all code addresses are known.
This is the normal situation for embedded software running on platforms with no operating system or a lightweight one;
it also applies to logical addresses%
\footnote{For systems with a memory management unit (MMU) and ``virtual memory'' we distinguish the logical addresses, as seen from the program, and the physical addresses, as seen from RAM.}
for non-relocatable programs running in full-fledged operating systems.
\item For data cache analysis, a points-to analysis has been run to obtain, for all memory accesses, a superset of the memory locations that may be affected.
In programs using pointer arithmetic or array accesses, this points-to analysis may need a value analysis.
Note that our analysis gives exact results for write-through write-allocate data caches, and that adaptation for other write policies is future work.
\item For mixed instruction/data caches, both of the above must have been done.
\item For caches addressed by physical addresses, the mapping from logical (virtual) to physical addresses must be known.%
\footnote{This is impossible for general operating systems with dynamic memory paging, but is possible for embedded systems: one may have a MMU to isolate processes from each other, e.g. high criticality from low criticality, with a fixed memory layout.}
\end{enumerate}

The result is a control-flow graph $G$ with vertices representing program locations.
An edge from $v_1$ to $v_2$ is labeled with the address of the memory block accessed when control steps from $v_1$ to $v_2$, or with $\noAccess$ if no access is made (e.g. we are analyzing the data cache and the instruction corresponding to the edge accesses no data).
Note that this is not, in general, the same thing as the address of the memory access: for instance, with 64-byte lines, an access to a byte at address 127 is considered to be an address to the line at address 64.
One memory access can extend over several cache lines: for instance, if instead of a byte we access a 4-byte word at address 127, on a processor allowing unaligned accesses, then we access successively two lines at addresses 64 and 128.
More generally, if several memory accesses are performed at the same program location, this location must be split into several sub-locations according to the order of the accesses.

Furthermore, this graph may actually be a multigraph, with several edges, labeled differently, between the same pair of vertices: if points-to analysis returns a set of several possible addresses for one access, there is one edge per address.

In our examples, for instance in \Autoref{fig:CFG}, lowercase letters $a$, $b$\dots denote such addresses.
$G$ has special start vertices labeled with either $\emptyset$, meaning that program execution is assumed to start with an empty cache, or $\top$, meaning that program execution may start with an arbitrary cache state (all legal combinations of memory blocks and empty lines are possible);
other classes of initial vertices may be added if needed.
All program executions must start at a start vertex.
Without loss of generality, we assume all vertices and edges to be reachable from start vertices, and the start vertices not to be endpoints of any edges.

In an LRU cache, as with almost all replacement policies, each cache set is treated independently.
One can thus analyze the behavior of the program completely independently on each cache set~$S$: $G$ is \emph{sliced} according to cache set $S$ by replacing each address not in $S$ by $\noAccess$ (see \Autoref{fig:CFG-slicing}).
{\semibold In the rest of the article, unless noted otherwise, we shall thus assume a single cache set, without loss of generality.}
For efficiency, an implementation may wish to collapse vertices and edges making $\noAccess$ transitions only.

An LRU cache encountering an access to memory block $a$ loads $a$ into the cache, for instance, after an access to $a$ the cache state of a 4-way LRU cache is $(a,x,y,z)$ where $a$ is the most recently used and $z$ the least recently used.
Two cases may occur:
\begin{inparadesc}
\item[(Hit)] $a$ is already in the cache and is ``rejuvenated'', i.e., cache-internal status bits are updated to record that $a$ is the most-recently-used memory block; for instance, from the cache state $(x,y,a,z)$ and access to $a$ leads to the cache state $(a,x,y,z)$ and from the cache state $(a,x,y,z)$ an access to $a$ does not change the cache state\footnote{In FIFO caches, there is no rejuvenation: a ``hit'' does not change the cache.}.
\item[(Miss)] $a$ is not in the cache and the ``oldest'' memory block is evicted from $a$'s cache set, for instance, from a cache state $(w,x,y,z)$ an access to $a$ evicts block $z$ and leads to the cache state $(a,w,x,y)$.
\end{inparadesc}

As a consequence, along a program execution $E$, assuming an initially empty cache, an access to $a$ is a hit \emph{if and only if}
 at most $\cacheWays-1$ distinct memory blocks have been accessed along $E$ since the last access to $a$ (several accesses to the same memory block count as one).

An edge is said to be ``always hit'' if all executions passing through this edge encounter a cache hit at this edge; otherwise it is said to be ``may miss''.
An edge is said to be ``always miss'' if all executions passing through this edge encounter a cache miss at this edge; otherwise it is said to be ``may hit''.
See \Autoref{fig:must_may_hit} for an example.
We shall propose two analyses, one for classifying edges as either ``always hit'' or ``may miss'', the other for classifying edges as ``always miss'' or ``may hit''.


%% file: motivating_example.tex
\section{Motivating Examples}
\label{sec:motivating_example}
\subsection{Age-based Analysis vs Precise Analysis}
\begin{figure}[h]
\begin{center}
\begin{tikzpicture}[->,auto,node distance=4em]
\node (start) [draw,diamond] { $\sigma_0: \emptyset$ };
\node (q1) [right of=start] { $\sigma_1$ };
\node (q2) [above right of=q1] { $\sigma_2$ };
\node (q3) [right of=q2] { $\sigma_3$ };
\node (q4) [right of=q3] { $\sigma_4$ };
\node (q5) [below of=q3] { $\sigma_5$ };
\node (q6) [below right of=q4] { $\sigma_6$ };
\node (q7) [right of=q6] { $\sigma_7$ };
\node (q8) [right of=q7] { $\sigma_8$ };
\node (q9) [above right of=q6] { $\sigma_9$ };
\node (q10) [right of=q9] { $\sigma_{10}$ };
\node (q11) [right of=q10] { $\sigma_{11}$ };

\path (start) edge node {$a$} (q1);
\path (q1) edge node {$c$} (q2);
\path (q2) edge node {$b$} (q3);
\path (q3) edge node {$d$} (q4);
\path (q6) edge node {$c$} (q7);
\path (q7) edge node {$a$} (q8);
\path (q4) edge node { $\noAccess$ } (q6);
\path (q1) edge node {$b$} (q5);
\path (q5) edge node { $\noAccess$ } (q6);
\path (q6) edge node {$a$} (q9);
\path (q9) edge node {$e$} (q10);
\path (q10) edge node {$c$} (q11);
\end{tikzpicture}
\end{center}
\caption{Example of control-flow graph where classical age-based ``must-hit'' analysis yields suboptimal results.}
\label{fig:cfg_suboptimal_age_must_hit}
\end{figure}
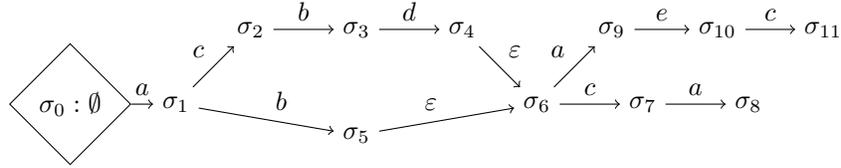

\begin{table}[h]
\caption{Concrete and abstract states in the analysis of the control-flow graph in \protect\Autoref{fig:cfg_suboptimal_age_must_hit}. To avoid too wide a table we omit the analysis results for $e$.
$\absent$ means ``absent''.}
\label{tab:cfg_suboptimal_age_must_hit}
\begin{center}
\footnotesize\setlength{\tabcolsep}{0.7ex}%
\resizebox{\textwidth}{!}{
\begin{tabular}{cccccccccc}
\toprule
  & Reachable cache states
  & \multicolumn{4}{c}{Age-based analysis}
  & \multicolumn{4}{c}{Block-focused analysis}\\
  \midrule
  && a & b & c & d & a & b & c & d\\
  \cmidrule(lr){3-6} \cmidrule(lr){7-10}
$\sigma_0$ & $(\noLine, \noLine, \noLine, \noLine)$
           & $\infty$ & $\infty$ & $\infty$ & $\infty$
           & $\absent$ & $\absent$ & $\absent$ & $\absent$ \\
$\sigma_1$ & $(a, \noLine, \noLine, \noLine)$
           & $0$ & $\infty$ & $\infty$ & $\infty$
           & $\emptyset$ & $\absent$ & $\absent$ & $\absent$ \\
$\sigma_2$ & $(c, a, \noLine, \noLine)$
           & $1$ & $\infty$ & $0$ & $\infty$
           & $\{ c \}$ & $\absent$ & $\emptyset$ & $\absent$ \\
$\sigma_3$ & $(b, c, a, \noLine)$
           & $2$ & $0$ & $1$ & $\infty$
           & $\{ b, c \}$ & $\emptyset$ & $\{ b \} $ & $\absent$ \\
$\sigma_4$ & $(d, b, c, a)$
           & $3$ & $1$ & $2$ & $0$
           & $\{ b, c, d \}$ & $\{ d\}$ & $\{ b, d \} $ & $\emptyset$ \\
$\sigma_5$ & $(b, a, \noLine, \noLine)$
          & $1$ & $0$ & $\infty$ & $\infty$
          & $\{ b \}$ & $\emptyset$ & $\absent$ & $\absent$ \\
$\sigma_6$ & $(d, b, c, a)$, $(b, a, \noLine, \noLine)$
          & $[1,3]$ & $[0,1]$ & $[2,\infty]$ & $[0,\infty]$
          & $\{ b, c, d \}$, $\{ b \}$ & $\{d\}$, $\emptyset$
          & $\{ b, d \}$, $\absent$ & $\emptyset$, $\absent$ \\
$\sigma_7$ & $(c, d, b, a)$, $(c, b, a, \noLine)$
           & $[2,\infty]$ & $[1,2]$ & $0$ & $[1,\infty]$
           & $\{ b, c, d\}$, $\{b, c\}$ & $\{ c, d\}$, $\{ c \}$
           & $\emptyset$ & $\{ c \}$, $\absent$ \\
$\sigma_8$ &$(a, c, d, b)$, $(a, c, b, \noLine)$
           & $0$ & $[2,3]$ & $1$ & $[2,\infty]$
           & $\emptyset$ & $\{ a, c, d \}$, $\{ a, c \}$
           & $\{ a \}$ & $\{ a, c\}$, $\absent$ \\
$\sigma_9$ & $(a, d, b, c)$, $(a, b, \noLine, \noLine)$
          & $0$ & $[1,2]$ & $[2,\infty]$ & $[1,\infty]$
          & $\emptyset$ & $\{a, d\}$, $\{ a \}$
          & $\{ a, b, d \}$, $\absent$ & $\{ a \}$, $\absent$ \\
$\sigma_{10}$ & $(e, a, d, b)$, $(e, a, b, \noLine)$
          & $1$ & $[2,3]$ & $[3,\infty]$ & $[2,\infty]$
          & $\{e \}$ & $\{a, d, e\}$, $\{ a, e \}$
          & $\absent$ & $\{ a, e \}$, $\absent$ \\
$\sigma_{11}$ & $(c, e, a, d)$, $(c, e, a, b)$
          & $2$ & $[3,\infty]$ & $0$ & $[3,\infty]$
          & $\{c, e \}$ & $\absent$, $\{ a, c, e \}$
          & $\absent$ & $\{ a, c, e \}$, $\absent$ \\
          \bottomrule
\end{tabular}}
\end{center}
\end{table}

To analyze LRU caches, it is convenient to introduce the following notion of the \emph{age} of a memory block:
The age of a block $x$ in a concrete cache state is the number of blocks younger than $x$, i.e., the number of blocks that have been accessed more recently than $x$, or $\infty$ if $x$ is not in the cache.
For instance, assuming a cache of associativity $4$, in the cache state $(c, b, a, \noLine)$, from youngest to oldest, $c$ has age $0$, $b$ has age $1$, $a$ has age $2$, $d$ has age $\infty$.

The classical age-based analysis \cite{DBLP:journals/rts/FerdinandW99} abstracts the set of possible cache states as follows: to each block $x$ it attaches a range of possible ages in the cache.
This abstraction may lead to imprecise results, that is, it may fail to conclude that an access is always a hit (respectively, a miss) whereas it is truly always a hit (respectively, a miss) in all executions.
Consider the control-flow graph in \Autoref{fig:cfg_suboptimal_age_must_hit}, and the corresponding concrete and abstract cache states in \Autoref{tab:cfg_suboptimal_age_must_hit}.
At $\sigma_4$, the age of $a$ is $3$ and at $\sigma_5$ it is $1$, so at $\sigma_6$ it is known to be in $[1,3]$;
similarly the age of $c$ is known to be in $[2,\infty]$.
When analyzing $\sigma_6 \xrightarrow{c} \sigma_7$, it is thus unknown whether $c$ is younger or older than $a$ in the cache at $\sigma_6$: in the former case $a$'s age does not change, whereas in the latter case $a$'s age increases by one.
The age-based analysis concludes that at $\sigma_7$, the age of $a$ is in $[2,\infty]$, whereas the exact range is $[2,3]$.
The age-based analysis thus cannot conclude that $\sigma_7 \xrightarrow{a} \sigma_8$ is a hit, which is the case in reality.

Instead of abstracting a concrete state with respect to a block $x$ by the \emph{number} of blocks younger than $x$, we abstract it by the \emph{set} of these blocks; thus a set of concrete states is abstracted by a set of sets of blocks.
For instance, at $\sigma_7$, we consider the possible sets of blocks younger than $a$ in the cache: $\{b, c, d \}$ and $\{b, c\}$.
This analysis is exact, in the sense that no precision is lost by performing analysis steps on the abstract states compared to performing the steps concretely and then abstracting the final result.

Note that the second set is included in the first. If $a$ is a hit after executing a sequence of steps from $\sigma_7$, from a cache state where $a$ is preceded by $\{b, c, d \}$, then \emph{a fortiori} the same sequence of steps also results in a hit if started in a cache state where $a$ is preceded by $\{b, c\}$.
We can thus discard $\{b, c, d \}$ from an analysis aimed at discovering the existence of hits (``may hit'' analysis).
Similarly, we can discard $\{b, c\}$ from an analysis aimed at discovering the existence of misses (the complement of an ``always hit'' analysis).
Again, doing this does not impact precision.

Similarly, in the age-based analysis, when analyzing $\sigma_6 \xrightarrow{a} \sigma_9$, it is unknown whether $c$ is younger or older than $a$ in the cache at $\sigma_6$: in the former case $c$'s age does not change, whereas in the latter case $c$ ages by one.
The age-based analysis concludes that at $\sigma_9$, the age of $c$ is in $[2,\infty]$, whereas the exact range is $[3,\infty]$.
The step $\sigma_9 \xrightarrow{e} \sigma_{10}$, with a fresh letter $e$, results in an increase of the ages of all other blocks.
The age-based analysis thus cannot conclude that $\sigma_{10} \xrightarrow{c} \sigma_{11}$ is a miss, which is the case in reality.

Note that the ``definitely-unknown'' abstract analysis proposed by \citet{DBLP:conf/cav/TouzeauMMR17} would not help in any way: it resolves some of the cases where there are execution traces leading both to a hit and a miss at the same location, which is not the case here.
Their approach would then have to call a model checker to establish that $\sigma_7 \xrightarrow{a} \sigma_8$ is a ``must hit'' and $\sigma_{10} \xrightarrow{c} \sigma_{11}$ a ``must miss''.
The purpose of our analysis is to replace this expensive call to a model checker by an abstract interpretation that yields the same result.

\subsection{Collecting Semantics vs Focused Semantics vs Antichain}

\begin{figure}[h]
\begin{center}
\begin{tikzpicture}[->,auto,node distance=4em]
\node (q0) [draw,diamond] { $\sigma_0: \emptyset$ };
\node (q1) [right of=q0] { $\sigma_1$ };
\node (q2) [right of=q1] { $\sigma_2$ };
\node (qdots) [right of=q2] { $\dots$ };
\node (qdotstwo) [right of=qdots] { $\dots$ };
\node (qn) [right of=qdotstwo] { $\sigma_n$ };
\node (qn1) [right of=qn] { $\sigma_{n+1}$ };
\path (q0) edge node { $a$ } (q1);
\path (q1) edge [bend left] node[above] { $b_1$ } (q2);
\path (q1) edge [bend right] node[below] { $\noAccess$ } (q2);
\path (q2) edge [bend left] node[above] { $b_2$ } (qdots);
\path (q2) edge [bend right] node[below] { $\noAccess$ } (qdots);
\path (qdots) edge [dotted,bend left] (qdotstwo);
\path (qdots) edge [dotted,bend right] (qdotstwo);
\path (qdotstwo) edge [bend left] node[above] { $b_{n-1}$ } (qn);
\path (qdotstwo) edge [bend right] node[below] { $\noAccess$ } (qn);
\path (qn) edge [bend left] node[above] { $b_n$ } (qn1);
\path (qn) edge [bend right] node[below] { $\noAccess$ } (qn1);
\end{tikzpicture}
\end{center}
\caption{Example where collecting and focused semantics are unnecessarily detailed and inefficient. $\noAccess$ means ``no access''.}
\label{fig:collecting_and_focused_are_inefficient}
\end{figure}
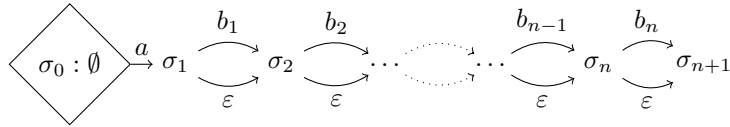

Consider the control-flow graph of \Autoref{fig:collecting_and_focused_are_inefficient}, with an associativity $\cacheWays > n$, and an empty initial cache.
At the last control location $\sigma_{n+1}$, the possible cache states are all subsequences of $b_n,\dots,b_1$ followed by $a$ and possibly empty lines, e.g. $(b_5,b_3,b_1,a, \noLine,\noLine)$. There are therefore $2^n$ reachable cache states at  $\sigma_{n+1}$, all of which appear in the collecting semantics of the program composed with the cache.

The ``block-focused'' abstraction, which was also applied by \citet{DBLP:conf/cav/TouzeauMMR17} when encoding cache problems into model-checking reachability problems, records only the set of blocks younger than the block of interest.
Here, if our block of interest is $a$, this abstraction thus yields at $\sigma_{n+1}$ the set of subsets of $\{ b_1, \dots, b_n \}$.
Of course, symbolic set representation techniques may have a compact representation for such a set, but the main issue is that this set keeps too much information.

If our goal is to prove the existence of a ``hit'' on an access to memory block $a$ further down the execution, then it is sufficient to keep, in this set, $\emptyset$, corresponding to a path composed of the access to $a$ followed by a sequence of no accesses
$\noAccess$.
More generally, it is sufficient to keep only the minimal elements (with respect to the inclusion ordering) from this set, which can be exponentially more succinct, as in this example.
Similarly, if our goal is to prove the existence of a ``miss'' on $a$ further down the execution, then it is sufficient to keep, in this set, $\{ b_1, \dots, b_n \}$, corresponding to a path composed of the access to $a$ followed by a sequence of accesses to $b_1$, \dots, $b_n$.
More generally, it is sufficient to keep only the maximal elements from this set.

This is the main difference between our analysis and the ``focused'' model that \citet{DBLP:conf/cav/TouzeauMMR17} fed into the model checker: the ``focused'' model contains unnecessary information (non-minimal elements for the Always-miss analysis, non-maximal elements for the Always-hit analysis), which increases model-checking times.
We discard these in our analysis.

The following section formally defines the collecting and focused semantics, and our new Always-hit and Always-miss analyses.
The Always-hit analysis (respectively Always-miss) computes the antichain of maximal elements (respectively minimal elements) i.e. the downward (respectively upward) closure of reachable states.


%% file: fixpoint.tex
\section{Analyses as Fixed-Point Problems}
\label{sec:fixpoint}
\subsection{Collecting Semantics}
To each vertex $l$ we attach a set $C_l$ of possible cache states.
Each cache state $s$ is a sequence $(s_1,\dots,s_{\cacheWays})$ of addresses, from youngest to oldest, possibly ending with one or more special values $\noLine$, meaning that this cache line is empty, and without repetition of addresses except for $\noLine$.
Let $S$ be the set of cache states.

\begin{example}
If $\cacheWays=4$, $(\noLine,\noLine,\noLine,\noLine)$ (empty cache), $(a,b,\noLine,\noLine)$, and $(a,d,c,b)$ are valid cache states;
$(a,b,\noLine,d)$ and $(a,d,a,b)$ are not.
\end{example}

To any start vertex $l$ labeled $\emptyset$, we attach $C_l = \{ (\noLine,\dots, \noLine) \}$, meaning that the only possible cache state at this location is empty.
To any start vertex $l$ labeled $\top$, we attach $C_l = S$,  meaning that any cache state is possible at this location.
The rest of the cache states are obtained as the least solution of:
\begin{description}
\item[Miss:] if $(s_1,\dots,s_{\cacheWays}) \in C_l$, none of the $s_i$ is $a$, and there is an edge $l \xrightarrow{a} l'$, then $(a, s_1,\dots,s_{\cacheWays-1}) \in C_{l'}$: the oldest line is evicted;
\item[Hit:] if $(s_1,\dots,s_{\cacheWays}) \in C_l$, $a$ occurs at position $i$ ($s_i = a$), and there is an edge $l \xrightarrow{a} l'$, then $(a, s_1,\dots,s_{i-1},s_{i+1},\dots,s_{\cacheWays}) \in C_{l'}$ (if $i=\cacheWays$, that is $(a, s_1,\dots,s_{i-1})\in C_{l'}$): block $a$ is ``rejuvenated''.
\end{description}

\begin{example}
  If the cache contains $(a,b,c,d)$ and $b$ is accessed, then $b$ is rejuvenated and the cache then contains $(b,a,c,d)$. If instead $e$ is accessed, then $d$ is evicted and the cache then contains $(e,a,b,c)$.
\end{example}

\subsection{Focused Semantics}
We reuse the ``focused semantics'' proposed by \citet{DBLP:conf/cav/TouzeauMMR17}, as well as their proof of exactness.

Let $a \in \locations$ be some address; we are interested in classifying accesses to $a$.
We can focus the behavior of the cache with respect to $a$ as follows.
A cache state $s_1,\dots,s_{\cacheWays}$ such that $s_i = a$ will be abstracted as the set $\{ s_1,\dots, s_{i-1} \}$ ($\emptyset$ if $i=1$) of blocks present in the cache before~$a$, i.e., younger than~$a$.
A cache state not containing $a$ is abstracted as the special value~$\absent$.

\begin{example}
When focusing on block $a$, cache states $(c,b,a,\noLine)$ and $(b,c,a,d)$ are both abstracted as the set $\{b,c\}$,
and $(c,b,e,\noLine)$ as~$\absent$.
\end{example}

One can easily show that this focused semantics can be directly computed as follows.
To each vertex~$l$ we attach a set $C_{l,a}$ of $a$-focused states.
These sets are ordered by inclusion.
To any start vertex $l$ labeled $\emptyset$, we attach $C_{l,a} = \{ \absent \}$.
To any start vertex $l$ labeled $\top$, we attach
$C_{l,a} = \{ \absent \} \cup
          \{ S \mid S \subseteq \locations \setminus \{ a \}
                    \land |S| \leq \cacheWays-1 \}$.
The rest of the cache states are obtained as the least solution of:
\begin{itemize}
\item if there is an edge $l \xrightarrow{a} l'$,
  then $\emptyset \in C_{l',a}$.
\item if $\absent \in C_{l,a}$ and there is an edge $l \xrightarrow{b} l'$,
  $b \neq a$ then $\absent \in C_{l',a}$;
\item if $S \in C_{l,a}$, $S \neq \absent$, $|S \cup \{b\}| < \cacheWays$
	and there is an edge $l \xrightarrow{b} l'$,
	$b \neq a$ then $S \cup \{ b \} \in C_{l',a}$;
\item if $S \in C_{l,a}$, $S \neq \absent$, $|S \cup \{b\}| = \cacheWays$
	and there is an edge $l \xrightarrow{b} l'$,
	$b \neq a$ then $\absent \in C_{l',a}$;
\end{itemize}

An edge $l \xrightarrow{a} l'$ may result in a miss if and only if
$\absent \in C_{l,a}$.
An edge $l \xrightarrow{a} l'$ may result in a hit if and only if
there exists $S \in C_{l,a}$,  $S \neq \absent$.

Another intuitive characterization is that $C_{l,a}$ is the collection of the sets of addresses found along the paths from the nearest preceding occurrences of~$a$, truncated at associativity;
sets of cardinality greater than or equal to associativity are all abstracted to the special value~$\absent$.

\subsection{Always-Hit Analysis}
\label{sec:may-miss}
To get a better idea of what the Always-hit analysis computes let us first recall the definitions of antichain and upper set, and illustrate this analysis with an example.
\begin{definition}
  An \emph{antichain} is a subset of an ordered set such that no two distinct elements of that subset are comparable.
  An \emph{upper set} (respectively \emph{lower set}) is a set such that if an element is in this set, then all elements larger (respectively, smaller) than it are also in the set.
\end{definition}

\begin{example}
Assume that a node $l$ may be reached with cache states
$C_l = \{ (c,b,e,a)$, $(b,c,d,a)$, $(b,a,\noLine,\noLine) \}$.
The possible a-focused states are:
$C_{l,a} = \left\{ \{ b, c, e \}, \{ b, c, d \}, \{ b \} \right\}$.
Since $\{ b \}$ is strictly included in $\{ b, c, e \}$ (and in $\{ b, c, d \}$), it may not contribute to cache misses that would not also occur following $\{ b, c, e \}$ (and $\{ b, c, d \}$) and can be removed without affecting soundness;
the antichain of the maximal elements of $C_{l,a}$ is $\left\{ \{ b, c, e \}, \{ b, c, d \} \right\}$, which will be called $C^{\max}_{l,a}$.
\end{example}

In all that follows, the ordering will be the inclusion ordering~$\subset$.

Recall that $S \in C_{l,a}, S \neq \absent$ means that at position $l$, there is a reachable cache state of the form $(s_1,\dots,s_{|S|},a,\dots)$ where $S = \{  s_1,\dots,s_{|S|} \}$.
An edge $l \xrightarrow{a} l'$ ``always hits'' if and only if it ``may not miss'', that is, if there is no execution trace leading to a miss at this location, i.e. $\absent \notin C_{l,a}$.

\begin{definition}
  Let $x \xrightarrow{b} y$ denote the transition  ``upon an access to block $b$, $b \neq a$, the cache may move from an $a$-focused state $x$ to an $a$-focused state $y$''.
Recall that $x$ may be $\absent$ ($a$ is not in the cache) or a subset of cache blocks, not containing $a$, of cardinality at most $\cacheWays - 1$.
This deterministic transition relation is defined as follows:
\begin{itemize}
\item $\absent \xrightarrow{b} \absent$;
\item $x \xrightarrow{b} x \cup \{ b \}$ for $|x \cup \{b\}| < \cacheWays $;
\item $x \xrightarrow{b} \absent$ for $|x \cup \{b\}| =\cacheWays$.
\end{itemize}
\end{definition}

\begin{definition}
  For $x,y \subseteq \locations \setminus \{ a \}$,
  let $x \xrightarrow{\downarrow} y$ denote a downward closure step:
\begin{itemize}
\item $\absent \xrightarrow{\downarrow} y$ for any $y$;
\item $x \xrightarrow{\downarrow} y$ for any $x,y \neq \absent$, $y \subseteq x$.
\end{itemize}
\end{definition}

\begin{example}
  $\{b, c, e\} \xrightarrow{\downarrow} y$ for any $y \in \{\emptyset, \{b\}, \{c\}, \{e\}, \{b, c\}, \{b, e\}, \{c, e\}, \{b, c, e\}\}$
\end{example}

\begin{lemma}
  Assume there are $x,y, z$, and $b$ such that $x \xrightarrow{\downarrow} y  \xrightarrow{b} z$.
  Then there exists $y'$ such that $x \xrightarrow{b} y' \xrightarrow{\downarrow} z$.
\end{lemma}
\begin{proof}
  As the transition relation is deterministic, $y'$ is uniquely determined by $x$ and $b$.

We distinguish two cases based on the value of $y'$:
\begin{enumerate}
	\item If $y' = \absent$, then the results follows immediately, as $\absent \xrightarrow{\downarrow} z$ for any $z$.
	\item If $y' \neq \absent$, then $y' = x \cup \{b\}$ with $|y'| < N$. Then $x \neq \absent$ and $y \subseteq x$, and so $z = y \cup \{b\} \subseteq x \cup \{b\} = y'$.
\end{enumerate}
\vspace{-5mm}
\end{proof}

\begin{corollary}
  There exists a sequence of the form
  $u_0 \xrightarrow{b_0} v_0 \xrightarrow{\downarrow}
  u_1 \xrightarrow{b_1} v_1 \xrightarrow{\downarrow} \dots
  u_n \xrightarrow{b_n} v_n  \xrightarrow{\downarrow} u_{n+1}$
  if and only if there exists a sequence of the form
  $u_0 \xrightarrow{b_0} u'_1
       \xrightarrow{b_1} u'_2 \dots
       \xrightarrow{b_n} u'_n \xrightarrow{\downarrow} u_{n+1}$.
\end{corollary}

It is thus equivalent to compute the reachable states of the $a$-focused semantics (for transitions different from $a$), then apply downward closure, and to apply downward closure at every step during the computation of reachable states.
In addition, it is obvious that $\absent$ is in a set if and only if it is in its downward closure.
It is thus equivalent to test for a ``may miss'' on the reachable states of the $a$-focused semantics and on their downward closure.

This suggests two possible (and equivalent, in a sense) simplifications to the focused semantics if our goal is to find places where an access to $a$ may be a miss:
\begin{description}
\item[Closure] Replace $C_{l,a}$ by its down-closure $C^{\downarrow}_{l,a}$:
  $S' \in C^{\downarrow}_{l,a}$ if and only if there exists
  $S \in C_{l,a}$ such that $S' \subseteq S$.
\item[Subsumption removal] Replace  $C_{l,a}$ by the antichain of its
  maximal elements:
  $S \in C^{\max}_{l,a}$ if and only if
  $S \in C_{l,a}$ and there is no $S' \in C_{l,a}$
  such that $S \subsetneq S'$.
\end{description}
Note that $C^{\downarrow}_{l,a}$ is the down-closure of $C^{\max}_{l,a}$, and
that $C^{\max}_{l,a}$ is the antichain of maximal elements of $C^{\downarrow}_{l,a}$;
thus $C^{\max}_{l,a}$ is just an alternative representation for $C^{\downarrow}_{l,a}$.
Our idea is to directly compute~$C^{\max}_{l,a}$.


\subsection{Always-Miss Analysis}
\label{sec:may-hit}

This subsection presents the Always-miss analysis which is the dual of the Always-hit analysis of Section~\ref{sec:may-miss}.
A control location ``always misses'' if and only if it ``may not hit'', that is, if there is no execution trace leading to a hit at this location.

\begin{definition}
  For $x,y \subseteq \locations \setminus \{ a \}$,
  let $x \xrightarrow{\uparrow} y$ denote an upward closure step:
\begin{itemize}
\item $x \xrightarrow{\uparrow} \absent$ for any $x$;
\item $x \xrightarrow{\uparrow} y$ for any $x,y \neq \absent$, $x \subseteq y$.
\end{itemize}
\end{definition}

\begin{lemma}
  Assume there are $x,y, z$, and $b$ such that $x \xrightarrow{\uparrow} y  \xrightarrow{b} z$.
  Then there exists $y'$ such that $x \xrightarrow{b} y' \xrightarrow{\uparrow} z$.
\end{lemma}
\begin{proof}
We distinguish two cases based on the value of $z$:
\begin{enumerate}
	\item If $z = \absent$, then the results follows immediately, as $y' \xrightarrow{\uparrow} \absent$ for any $y'$.
	\item If $z \neq \absent$, then $z = y \cup \{b\}$ with $|z| < N$. Then $y \neq \absent$ and $x \subseteq y$, and so $y' = x \cup \{b\} \subseteq y \cup \{b\} = z$.
\end{enumerate}
\vspace{-5mm}
\end{proof}

%
%
%




\begin{corollary}
  There exists a sequence of the form
  $u_0 \xrightarrow{b_0} v_0 \xrightarrow{\uparrow}
  u_1 \xrightarrow{b_1} v_1 \xrightarrow{\uparrow} \dots
  u_n \xrightarrow{b_n} v_n  \xrightarrow{\uparrow} u_{n+1}$
  if and only if there exists a sequence of the form
  $u_0 \xrightarrow{b_0} u'_1
       \xrightarrow{b_1} u'_2 \dots
       \xrightarrow{b_n} u'_n \xrightarrow{\uparrow} u_{n+1}$.
\end{corollary}

It is thus equivalent to compute the reachable states of the $a$-focused semantics (for transitions different from $a$), then apply upward closure, and to apply upward closure at every step during the computation of reachable states.
In addition, it is obvious that there exists $x$ in $X$, $x \neq \absent$, if and only if there exists $y$ in the upward closure of $X$ such that $y \neq \absent$.
It is thus equivalent to test for a ``may hit'' on the reachable states of the $a$-focused semantics and on their upward closure.

This again suggests two possible simplifications to the focused semantics if our goal is to find places where an access to $a$ may be a hit:
\begin{description}
\item[Closure] Replace $C_{l,a}$ by its up-closure $C^{\uparrow}_{l,a}$:
  $S' \in C^{\uparrow}_{l,a}$ if and only if there exists
  $S \in C_{l,a}$ such that $S \subseteq S'$.
\item[Subsumption removal] Replace  $C_{l,a}$ by the antichain of its
  minimal elements:
  $S \in C^{\min}_{l,a}$ if and only if
  $S \in C_{l,a}$ and there is no $S' \in C_{l,a}$
  such that $S' \subsetneq S$.
\end{description}
Note that $C^{\uparrow}_{l,a}$ is the up-closure of $C^{\min}_{l,a}$, and
that $C^{\min}_{l,a}$ is the antichain of minimal elements of $C^{\uparrow}_{l,a}$;
thus $C^{\min}_{l,a}$ is just an alternative representation for $C^{\uparrow}_{l,a}$.
Our idea is to directly compute~$C^{\min}_{l,a}$.

\subsection{A Remark on Lattice Height}
We replace the focused semantics by its upward or downward closure;
this is a form of convergence acceleration, albeit one that preserves the precision of the final result.
We shall see in \Autoref{sec:implementation} that this improves practical performance considerably compared to a version that checks the focused semantics in a model checker.
It is however unlikely that this improvement translates to the worst case; let us see why.

The number of iterations of a data-flow or abstract interpretation analysis is bounded by the height of the analysis lattice, that is, the maximal length of a strictly increasing sequence.
However, this height does not change by imposing that the sets should be lower (respectively upper) closed: just apply the following lemma to $T$, the set of subsets of $\locations$ of cardinality at most $\cacheWays-1$ (plus $\absent$) ordered by inclusion (respectively, reverse inclusion).

\begin{lemma}
  Let $(T,\leq)$ be a partially ordered finite set.
  The lattice of lower subsets of $T$, ordered by inclusion, has height~$|T|$, the same height as the lattice of subsets of $T$.
\end{lemma}

\begin{proof}
  Order $T$ topologically: $t_1,\dots,t_{|T|}$, such that
  $\forall i,j:~ t_i \leq t_j \implies i \leq j$. 
  The sequence $(u_i)_{i=0,\dots, |T|}$, with $u_i = \{ t_1, \dots, t_i \}$, is a strictly ascending sequence of lower sets.
\end{proof}


%% file: zdd.tex
\section{Data Structures and Algorithms}
\label{sec:zdd}

In \Autoref{sec:fixpoint} we defined a collecting semantics for concrete cache states, then, in two steps (1. focused semantics, 2. closures), showed that there is a cache hit (respectively, a cache miss) in the concrete semantics if and only if there is a cache hit (respectively, a cache miss) in an upward-closed (respectively, downward-closed) semantics, and that upward-closed (respectively, downward-closed) sets may be represented by the antichains of their minimal (respectively, maximal) elements.

\subsection{Computation by Abstract Interpretation}

The abstracted semantics in upward-closed (or, downward-closed) sets may be computed by a standard data-flow/abstract interpretation algorithm, by upward iterations, as follows.

To each initial control point we initially attach an initialization value (see below).
For a semantics focused on accesses to $a$, we consider that each edge $x \xrightarrow{a} y$ is replaced by an initial edge $\emptyset \xrightarrow{a} y$, pushing $\{ \emptyset \}$ as the value associated to the control state~$y$.
Then, we iterate in the usual abstract interpretation fashion: we maintain a ``working set'', initially containing the initial locations and the targets of the $\xrightarrow{a}$ edge; we take a control point $x$ from the working set, update the abstract values at the end point of edges going out of $x$ (using the union operation on upper or lower sets), and add these end points to the working set if their value has changed (equality testing). The iterations stop when the working set becomes empty.
It is a classical result \cite[\S 2.9]{Cousot78} that the final result of such iterations does not depend on the iteration ordering, and in fact several elements from the working set may be treated in parallel;
the only requirement is that all elements from the working set are eventually treated.

The sequence of updates to the set decorating a given control location is strictly ascending, in a finite lattice; thus its length is bounded by the height $h$ of that lattice.
If $V$ is the set of control locations, then the total number of updates is bounded by~$|V| \cdot h$.
Recall that the height of the lattice of subsets of a set $X$ is~$|X|$.

If we implement the focused semantics directly, then we compute over sets of subsets of size at most $\cacheWays-1$ of $\locations \setminus \{ a \}$, completed with $\absent$;
the number of such subsets is bounded by $\sum_{k=0}^{\cacheWays-1} (|\locations|-1)^k$
and thus $h \leq \frac{(|\locations|-1) ^ \cacheWays}{|\locations|-2} + 1$.
The cost could thus be exponential in the associativity;
we shall see in \Autoref{sec:complexity} that a polynomial-time algorithm is unlikely, since the problems are NP-complete.

\subsection{Closed Sets Implementation}

We initially attempted adding closure steps to the focused semantics, and running a model checker on the resulting systems.
The performance was however disappointing, worse than model-checking the focused semantics itself as was proposed by \citet{DBLP:conf/cav/TouzeauMMR17}.
The model checker (\software{nuXmv}) was representing its sets of sets of blocks using state-of-the-art binary decision diagrams;
we thus did not expect any gain by going to our own implementation of iterations over the same structure.
We thus moved from representing a closed set by its content to representing it by the antichain of its minimal (respectively, maximal) elements.
There remains the question of how to store and compute upon the antichains representing those sets.

We then tried storing an antichain simply as a sorted set of subsets of $\locations$, each subset being represented as the list of its elements.
Experimentally, this approach was inefficient; let us explain why, algorithmically.
For once, when computing the antichain for the union of two upward or downward closed sets $S$ and $S'$, one takes the antichains $W$ and $W'$ representing $S$ and $S'$ and eliminates redundancies; if such a naive representation is used, one needs to enumerate all pairs of items from $W \times W'$ ---
there is no way to immediately identify which parts of $W$ and $W'$ are subsumed, or even to identify which parts are identical.
Furthermore, there is no sharing of representation between related antichains.

Binary decision diagrams are one well-known data structure for representations of sets of states; they share identical subsets, and allow fast equality testing. All operations over such diagrams can be ``memoized'', meaning that when an operation is run twice between identical subparts of existing diagrams, the result may be cached.
We store an antichain, a set $S$ of sets of addresses, as a \emph{zero-suppressed decision diagram} (ZDD)~\cite{DBLP:journals/sttt/Minato01,Mishchenko_ZDD_2014} \cite[\S7.1.4, p.249]{TAOCP_4A}, a variant of binary decision diagrams optimized for representing sets of sparse sets of items.

\subsection{Basic Functions for May-Hit and May-Miss Analyses}

We assume that all control states are reachable (unreachable states are easily discarded by a graph traversal). The starting points of the analyses focused on $a$ are the initial control points as well as all accesses to~$a$.

The operations that we need for antichains defining upper sets, for the may-hit analysis, are
\begin{description}
\item[Initialization to empty cache]
  Return $\emptyset$.
\item[Initialization to undefined cache state]
  Return $\{ \emptyset \}$.
\item[Initialization to unreachable state]
  Return $\emptyset$.
\item[Access] to address $b \neq a$: return $\{ s \cup \{ b \} \mid s \in S \}$.
\item[Access to tracked block] $a$: return $\{ \emptyset \}$
\item[Limitation to associativity]
  Return $\{ s \mid s \in S \land |s| \leq \cacheWays - 1 \}$.%
\footnote{This means that execution traces that cannot lead to a ``hit'' on the next access to $a$ are discarded. This is correct since execution is assumed to start from all accesses to $a$ as well as initial control states.
If one wishes to combine the analysis with others which need to distinguish between ``no hit at the next access to $a$'' and ``unreachable'', the special value $\absent$ may be added when elements of too large associativity are discarded.}
\item[Union of upper sets] represented by antichain of minimal elements of $S$ and $S'$:
  return $\{ s \mid s \in S \land \neg\exists s' \in S'~ s' \subsetneq s \} \cup
  \{ s' \mid s' \in S' \land \neg\exists s \in S~ s \subsetneq s' \}$.
\item[Equality testing] given $S$ and $S'$, return whether $S = S'$.
\end{description}

\begin{example}
  Let $S$ be the upper set generated by the antichain $\big\{ \{a\}, \{b,c\} \big\}$, and
  $S'$ the upper set generated by the antichain $\big \{ \{b\}, \{a,c\}, \{d\} \big\}$.
  The union of the two upper sets is an up-set generated by the union of these two antichains.
  However, this union is not an antichain because it contains redundant items:
  $\{ a,c \}$ is subsumed by $\{ a \}$,
  $\{ b,c \}$ is subsumed by $\{ b \}$.
  The antichain of minimal elements of $S \cup S'$ is thus
  $\big\{ \{a\}, \{b\}, \{d\} \big\}$.
\end{example}

The operations that we need for antichains defining lower sets, for the may-miss analysis, are
\begin{description}
\item[Initialization to empty or undefined cache state]
  Return $\{ \absent \}$. 
\item[Initialization to unreachable state]
  Return $\emptyset$.
\item[Accesses] Same as with upper sets.
\item[Test for eviction]
  Returns whether there exists $s \in S$ such that $|s| \geq \cacheWays$, in which case $S$ is replaced by $\{\absent\}$ (again, this is an optimization).
\item[Union of lower sets]
  represented by antichains of maximal elements $S$ and $S'$:
  return $\{ s \mid s \in S \land \neg\exists s' \in S'~ s \subsetneq s' \} \cup
  \{ s' \mid s' \in S' \land \neg\exists s \in S~ s' \subsetneq s \}$.
\item[Equality testing] Same as with upper sets.
\end{description}

The union of antichains with subsumption removal was supported by an extension~\cite{Mishchenko_ZDD_2014} of the ZDD library that we used.
The only operations not supported were the test for eviction and the limitation to associativity.
We implemented them by recursive descent over the structure of the ZDD, with an extra parameter for the current depth (number of items already seen in the set), and memoization of the results.
As in the \software{Cudd} library, we call ``then'' the branch where the top variable is true (i.e. the branch that contains the cache block associated to the current node) and ``else'' the branch associated to value false (i.e. the branch that does not contain that cache block).
As shown in Algorithm~\ref{alg:truncate}, the general case (case 3) of the algorithm simply consists in truncating the ``then'' and ``else'' branches of the current nodes.
When the number of ``then'' branches taken reaches the associativity (case 2), we remove all further ``then'' branches (they only lead to sets of cardinality greater than the associativity).
Finally, the algorithm may stop exploring a branch for two different reasons:
\begin{inparaenum}[a)]
	\item either the node treated is a leaf of the ZDD (case 0), or
	\item the result of the truncate function has already been computed and memoized (case 1).
\end{inparaenum}

\newcommand{\zdd}{\textit{zdd}}
\newcommand{\rightSub}{\textit{else}}
\newcommand{\leftSub}{\textit{then}}

\begin{algorithm}[h]
\caption{Truncate(zdd, n) as a recursive function}\label{alg:truncate}
\begin{algorithmic}[1]
\Function{Truncate}{$zdd$, $n$}
\If {$zdd = \emptyset$ or $zdd = \{\emptyset\}$}
	\State \Return $zdd$ \Comment{Case 0. Leaf of the ZDD DAG}
\EndIf
\State $res \gets \Call{cacheLookup}{Truncate, \zdd, n}$ \Comment{Case 1. Already computed}
\If {$res$}
	\State \Return $res$
\EndIf
\If {$n = 0$} \Comment{Case 2. Associativity is reached}
	\State \Return \Call{Truncate}{$\zdd.\rightSub, 0$} \Comment{Case 2. Else branch recursion}
\Else \Comment{Case 3. General case}
	\State $\leftSub \gets \Call{Truncate}{\zdd.\leftSub, n-1}$ \Comment{Case 3. Then branch recursion}
	\State $\rightSub \gets \Call{Truncate}{\zdd.\rightSub, n}$ \Comment{Case 3. Else branch recursion}
	\State \Return \Call{ZDD}{$\zdd.var$, \leftSub, \rightSub}
\EndIf
\EndFunction
\end{algorithmic}
\end{algorithm}


%% file: extensions.tex
\section{Variants and Extensions}
\label{sec:extensions}
\paragraph{Combination with classical abstract interpretation}
When classical abstract interpretation \cite{DBLP:journals/rts/FerdinandW99}, or its combination with the ``definitely unknown'' abstract analysis \cite{DBLP:conf/cav/TouzeauMMR17}, can correctly classify all accesses to a given block $a$ into ``always hit'', ``always miss'' and ``definitely unknown'', there is no use in running our analysis for that block.
We have implemented this combination, which improves performance (see \Autoref{sec:implementation}).

\paragraph{Simultaneous computation}
We have explained our analyses for classifying accesses to each address $a$ separately. It is also possible to simultaneously classify all addresses together, by updating the abstractions (e.g. $C^{\min}_{l,a}$) for all $a$ all together when updating the abstract state at location~$l$.

This simultaneous computation, including across cache sets, is likely to be compulsory if the cache analysis is integrated with a microarchitectural analysis: if the sequence of memory accesses depends on whether some previous accesses are hits or misses, e.g. due to out-of-order execution or opportunistic prefetching.

\paragraph{On-demand backward analysis}
We have presented our analysis in a forward fashion: to classify hits and misses to $a$, we compute at each location the collection of the set of addresses found along path $\pi$ for all paths $\pi$ from the nearest preceding occurrences of~$a$ (truncated at length~$\cacheWays$).
We could formulate our analysis in a backward fashion: given a specific location $l$ in the control-flow graph, we compute at each location $l'$ the collection of the set of addresses found along path $\pi$ for all paths $\pi$ from $l'$ to $l$.
This computation stops at other edges labeled with $a$, start vertices, or when computing the special value~$\absent$.
Then, an edge going out of $l$ and labeled by $a$ may result in a miss if and only if at least one value $\absent$ or an $\emptyset$ start vertex was reached during this backward propagation, and it may result in a hit if and only if at least one edge $a$ or a $\top$ start vertex was reached during this backward propagation.


%% file: complexity.tex
\section{Complexity and NP-Hardness}
\label{sec:complexity}
The cache contains at most $\cacheWays$ cache lines chosen among $|\locations|$ memory blocks; the number of cache states is thus bounded by
$\sum_{k=0}^\cacheWays |\locations|^k = \frac{|\locations|^{\cacheWays+1}-1}{|\locations|-1}$.
Consequently, the total number of program states is bounded by $|V| \frac{|\locations|^{\cacheWays+1}-1}{|\locations|-1}$
where $V$ is the set of vertices.
Recall that $\locations$ is the set of possible addresses, which are used to label the edges $E$; thus $|\locations| \leq |E|$.
For a fixed associativity $\cacheWays$, an explicit model-checking approach, enumerating all cache states, thus has polynomial complexity in the size of the control-flow graph under analysis;
however its complexity is exponential in the cache's associativity.
Furthermore, for program analysis, any effective complexity beyond almost-linear in the size of the program is generally considered prohibitive.
This explains the development of abstract interpretation, with some imprecision \cite{DBLP:journals/rts/FerdinandW99}, as well as clever pre-analyses and model reductions before applying symbolic model checking \cite{DBLP:conf/cav/TouzeauMMR17}.
We shall now see that cache analysis problems for LRU caches are NP-hard, even if the control-flow graph is acyclic.

We here assume an empty initial cache. The \emph{may-hit} (respectively, \emph{may-miss}) problem is: given a control-flow multigraph and a designated edge $e$, does there exist a path through the graph such that the access on edge $e$ is a hit (respectively, a miss)?
The \emph{always hit} (respectively, \emph{always miss}) problem is its complement: is a given access in a control-flow graph always a hit (respectively, a miss) irrespective of how it is reached?

The input problem is given as
\begin{inparaenum}[(a)]
\item as in preceding sections, the control-flow multigraph, with edges labeled with the addresses of the data being accessed,
\item the designated edge to classify, and
\item the cache's associativity\footnote{Note that if the associativity is larger than the set of possible edge labels, the two problems reduce to simple reachability problems in a directed graph. We can thus assume the associativity to be less than the number of edge labels. Whether the associativity is written in unary or binary form then is of no importance for the complexity of the problem.}.%
\end{inparaenum}

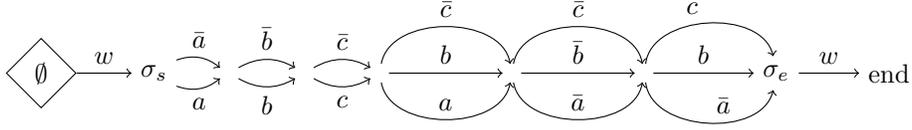
\begin{figure}
  \begin{center}
    \begin{tikzpicture}[->,auto]
      \node (start) [draw,diamond] { $\emptyset$ };

      \node (q0) [right of=start,xshift=5mm] { $\sigma_s$ };
      \node (q1) [right of=q0] { };
      \node (q2) [right of=q1] { };
      \node (q3) [right of=q2] { };

      \path (q0) edge[bend right] node[below] {$a$} (q1); 
      \path (q0) edge[bend left] node[above] {$\vneg{a}$} (q1); 

      \path (q1) edge[bend right] node[below] {$b$} (q2); 
      \path (q1) edge[bend left] node[above] {$\vneg{b}$} (q2); 

      \path (q2) edge[bend right] node[below] {$c$} (q3); 
      \path (q2) edge[bend left] node[above] {$\vneg{c}$} (q3);

      \node (q4) [right of=q3, node distance=5em] { };
      \node (q5) [right of=q4, node distance=5em] { };
      \node (q6) [right of=q5, node distance=5em] { $\sigma_e$  };

      \node (end) [right of=q6,xshift=5mm] { end };

      \path (start) edge node {$w$} (q0);
      \path (q6) edge node {$w$} (end);

      \path (q3) edge[bend right=80] node[auto] {$a$} (q4);
      \path (q3) edge node[auto] {$b$} (q4);
      \path (q3) edge[bend left=80] node[auto] {$\vneg{c}$} (q4);

      \path (q4) edge[bend right=80] node[auto] {$\vneg{a}$} (q5);
      \path (q4) edge node[auto] {$\vneg{b}$} (q5);
      \path (q4) edge[bend left=80] node[auto] {$\vneg{c}$} (q5);

      \path (q5) edge[bend right=80] node[auto] {$\vneg{a}$} (q6);
      \path (q5) edge node[auto] {$b$} (q6);
      \path (q5) edge[bend left=80] node[auto] {$c$} (q6);
    \end{tikzpicture}
  \end{center}

  \caption{Reduction from \protect\Autoref{th:exists-hit-acyclic-np-complete}.
There is a path from $\sigma_s$ to $\sigma_e$ with at most $\cacheWays-1=3$ different labels if and only if the formula $(\vneg{c} \lor b \lor a) \land (\vneg{c} \lor \vneg{b} \lor \vneg{a}) \land (c \lor b \lor \vneg{a})$ has a model.
This is equivalent to the existence of a path with a ``hit'' at the last $w$ edge for associativity $\cacheWays=4$.}
  \label{fig:exists-hit-reduction}
\end{figure}

\begin{theorem}\label{th:exists-hit-acyclic-np-complete}
  The may-hit problem is NP-complete for acyclic control-flow graphs.
\end{theorem}

\begin{proof}
  Obviously, the problem is in NP: a path may be chosen nondeterministically, then checked in polynomial time.

  Now consider the following reduction from CNF-SAT (see \Autoref{fig:exists-hit-reduction} for an example).
  To each variable $v$ in the SAT problem we associate two memory blocks $v$ and $\vneg{v}$.
  The control-flow graph is a sequence of switches:
  \begin{itemize}
  \item For all variables $v$ in the SAT problem, a switch between two edges labeled with $v$ and $\vneg{v}$, respectively.
  \item For each clause in the SAT problem, a switch between edges labeled with the literals present in the clause.
  \end{itemize}
  Let $n$ be the number of variables in the SAT problem.
  Each path through the sequence of switches with at most $n$ different labels corresponds to a satisfying assignment.
  Such a path exists if and only if the input formula is satisfiable.

  Now add to the control-flow graph an incoming edge into the first node and an outgoing edge from the last node, both labeled with the same fresh letter~$w$.
  The outgoing edge is the designated edge to classify.
  If the associativity of the cache is $n+1$, then the final access to $w$ may be a hit if and only if the SAT problem is satisfiable.
\end{proof}

\begin{figure}
\begin{subfigure}{0.25\textwidth}
\begin{center}
    \begin{tikzpicture}[node distance=3em]
      \node (q0) { $v_0$ };
      \node (q1) [above right of=q0] { $v_1$ };
      \node (q3) [below right of=q1] { $v_3$ };
      \node (q2) [below right of=q0]  { $v_2$ };
      \path (q0) edge[thick] (q1);
      \path (q0) edge[thick] (q2);
      \path (q1) edge (q2);
      \path (q1) edge[thick] (q3);
      \path (q2) edge[thick] (q3);
    \end{tikzpicture}
\end{center}
\caption{Graph with (thick) Hamiltonian cycle}
\end{subfigure}
\begin{subfigure}{0.57\textwidth}
\begin{center}
\begin{tikzpicture}[->,node distance=3.5em,auto]
      \node (start) [draw,diamond] { $\emptyset$ };

      \node (q0s) [right of=start] { $v_0^0$ };

      \node (q2_1) [right of=q0s]  { $v_2^1$ };
      \node (q1_1) [above of=q2_1] { $v_1^1$ };
      \node (q3_1) [below of=q2_1] { $v_3^1$ };

      \node (q1_2) [right of=q1_1] { $v_1^2$ };
      \node (q2_2) [below of=q1_2]  { $v_2^2$ };
      \node (q3_2) [below of=q2_2] { $v_3^2$ };

      \node (q1_3) [right of=q1_2] { $v_1^3$ };
      \node (q2_3) [below of=q1_3]  { $v_2^3$ };
      \node (q3_3) [below of=q2_3] { $v_3^3$ };

      \node (q0e) [right of=q2_3] { $v_0^4$ };

      \node (end) [right of=q0e] { end };

      \path (q0s) edge[thick] (q1_1);
      \path (q0s) edge (q2_1);

      \path (q1_1) edge (q2_2);
      \path (q1_1) edge[thick] (q3_2);
      \path (q2_1) edge (q3_2);
      \path (q2_1) edge (q1_2);
      \path (q3_1) edge (q1_2);
      \path (q3_1) edge (q2_2);

      \path (q1_2) edge (q2_3);
      \path (q1_2) edge (q3_3);
      \path (q2_2) edge (q3_3);
      \path (q2_2) edge (q1_3);
      \path (q3_2) edge (q1_3);
      \path (q3_2) edge[thick] (q2_3);

      \path (q1_3) edge (q0e);
      \path (q2_3) edge[thick] (q0e);

      \path (start) edge[thick] node {$w$} (q0s);
      \path (q0e) edge[thick] node {$w$} (end);
\end{tikzpicture}
\end{center}
\caption{Acyclic control flow graph obtained by the reduction.
Edge labels not shown.
Path corresponding to the Hamiltonian cycle (thick).}
\end{subfigure}

\caption{Reduction from \protect\Autoref{th:exists-miss-acyclic-np-complete}.}
\label{fig:exists-miss-reduction}
\end{figure}

\begin{theorem}\label{th:exists-miss-acyclic-np-complete}
  The may-miss problem is NP-complete for acyclic control-flow graphs.
\end{theorem}

\begin{proof}
  Obviously, the problem is in NP: a path may be chosen nondeterministically, then checked in polynomial time.

  We reduce the Hamiltonian circuit problem to the may-miss problem (see \Autoref{fig:exists-miss-reduction} for an example).
  Let $(V,E)$ be a graph, let $n = |V|$, $v=\{v_0,\dots,v_{n-1}\}$
  (the ordering is arbitrary).
  Let us construct an acyclic control-flow graph $G$ suitable for cache analysis as follows:
  \begin{itemize}
    \item two copies $v_0^0$ and $v_0^n$ of $v_0$
    \item for each $v_i$, $i \geq 1$, $|V|-1=n-1$ copies $v_i^j$, $1 \leq j < n$
      (this arranges these vertices in layers indexed by~$j$)
    \item for each pair $v_i^j$, $v_{i'}^{j+1}$ of nodes in consecutive layers,
      an edge, labeled by the address~$i'$, if and only if there is
      an edge $(i,i')$ in~$E$.
  \end{itemize}
  There is a Hamiltonian circuit in $(V,E)$ if and only if there is a path in $G$ from $v_0^0$ to $v_0^n$ such that no edge label is repeated,
  thus if and only if there exists a path from $v_0^0$ to $v_0^n$ with at least $n$ distinct edge labels.

  Now assume an edge going from a start node into $v_0^0$, and an edge going from $v_0^n$ into an end node, both labeled with the same fresh letter~$w$.
  The edge going from $v_0^n$ is the designated edge to classify.
  For associativity $n$ there exists an access missing the cache at that last edge if and only if there is a path from $v_0^0$ to $v_0^n$ with at least $n$ distinct edge labels.
\end{proof}

We have shown in this section how to construct CFGs for which solving the exist-miss and exist-hit problems is hard.
Note that this implies that both problems are NP-complete in the general case, but not that there is no algorithm for efficiently dealing with ordinary CFGs.


%% file: implementation.tex
\section{Implementation and Experiments}
\label{sec:implementation}

We have implemented our antichain-based analysis, as well as the classical age-based analysis \cite{DBLP:journals/rts/FerdinandW99}, and the ``definitely unknown'' (DU) and ``focused'' model-checking analyses proposed by \citet{DBLP:conf/cav/TouzeauMMR17}.
We did not implement the naive collecting semantics approach (model checking with ``unfocused'' cache states) since \citet[\S 6.3] {DBLP:conf/cav/TouzeauMMR17} note that then the models become so large and complex that the model checker timed out on all of their examples.
Furthermore, initial experiments with concretely represented antichains (an antichain being represented as a concrete set of arrays of block identifiers) scaled very poorly, so we did not pursue that direction further and pursued a fully symbolic representation using ZDDs.

Our experiments are performed on a server with 64~GB of memory, and an Intel Xeon CPU E5-2650 (32~logical cores running at 2.0~GHz).
The tested implementation is fully sequential, and thus does not benefit from the high number of cores available.
Note that the approach could however easily be implemented in parallel, by analyzing a different cache block on each core\footnote{Using threads if the ZDD library is capable of dealing with one different ZDD manager per thread, or separate processes.}.
\oldnew{We analyze a 4~KB cache with 64 cache sets, 4 ways and cache lines holding 16-byte-sized memory blocks.}{We analyze a 4~KB cache with 32 cache sets, 8 ways\footnote{Note that associativities of 8 or even 16 are common in modern microarchitectures. For instance, in the AMD Ryzen microarchitecture~\cite{amdOptManual17}, the L1 data cache and the unified L2 cache are 8-way set-associative, while the shared L3 cache even consists of 16 ways. Similarly, in the Intel Skylake microarchitecture~\cite{intelOptManual16}, the L1 data and instruction caches are both 8-way set-associative and, depending on the specific model, the shared L3 cache consists of up to 16 ways.} and cache lines holding 16-byte-sized memory blocks.}

\begin{figure}
\begin{center}
\includegraphics[width=\textwidth]{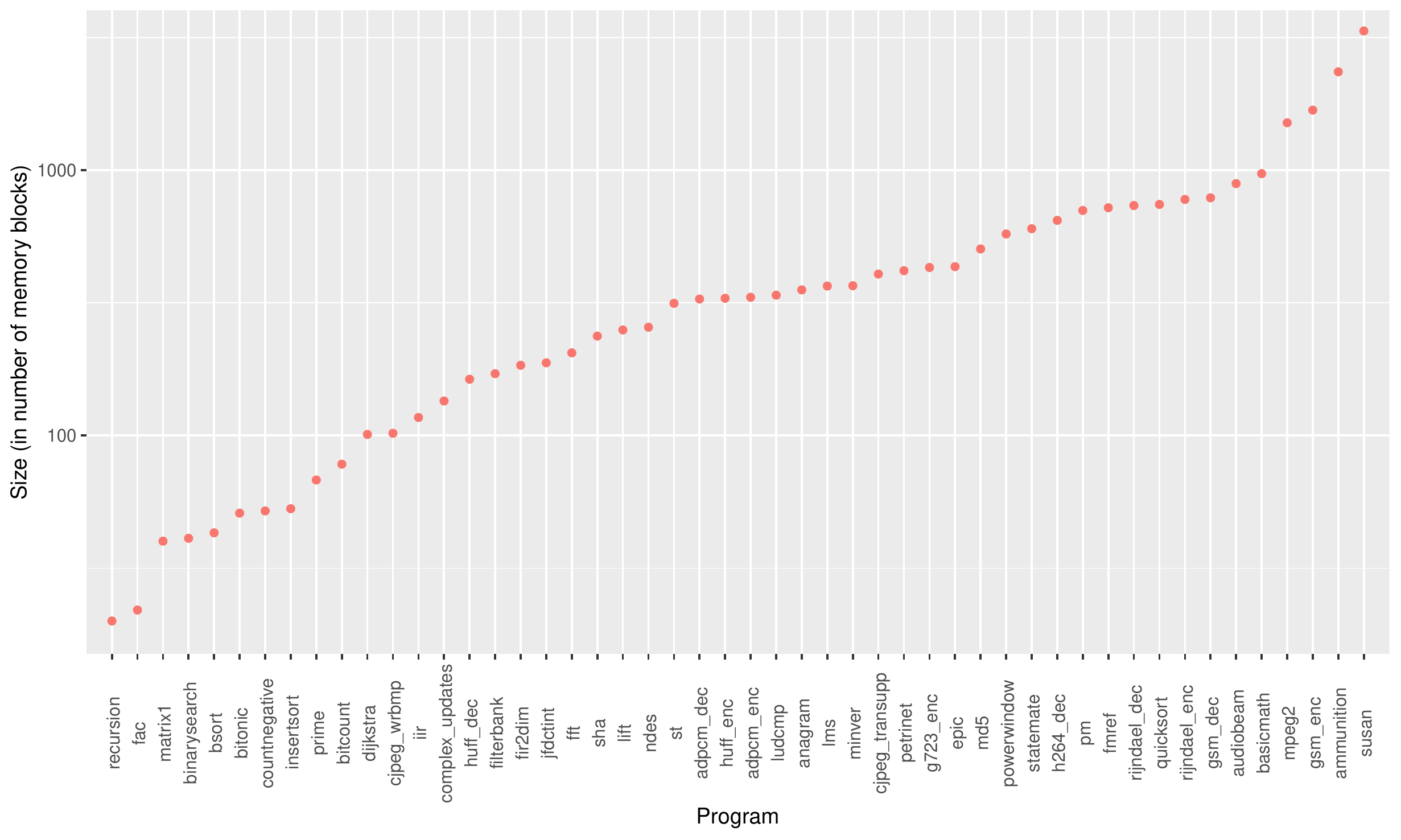}
\caption{Size of the analyzed binary code.}\label{fig:bench_size}
\end{center}
\end{figure}

We evaluate our approach on all sequential benchmarks, i.e., excluding parallel benchmarks, from the \software{TacleBench}\footnote{\software{TacleBench} is available at \url{https://github.com/tacle/tacle-bench.git}.}~\cite{FalkWCET2016}  suite, which is also used by \citet{DBLP:conf/cav/TouzeauMMR17}.
The benchmarks vary in size from 70 to 13000 lines of C code, and the size of the binary files obtained when compiling for ARM 5 (supported by \software{Otawa}) are shown in~\Autoref{fig:bench_size}.
Sizes are given in the number of memory blocks, and range from 20 blocks for the smallest benchmark to 3348 blocks for the biggest benchmark.
When measuring the time and memory consumption of analyses, we use a timeout of 12 hours.
Consequently, when this timeout is reached (i.e. the approach did not finish classifying accesses in the available amount of time), \oldnew{the associated point is not shown in the figures}{the associated point is plotted as if the corresponding analysis had terminated after 12 hours}. 


We have implemented our analyses on top of \software{Otawa}~\cite{DBLP:conf/seus/BallabrigaCRS10}, an open-source WCET analysis tool.%
\footnote{We have used the version 2 obtained at \url{https://www.tracesgroup.net/otawa/download/otawa-v2/}.}
Computations over ZDDs are performed by \software{Cudd}~2.3.1 \cite{DBLP:journals/sttt/Somenzi01} together with an extension \cite{Mishchenko_ZDD_2014} for computing over antichains\footnote{This extension has not been ported to more recent versions of \software{Cudd}.}.
In order to compare our new analysis to the previous analysis by~\citet{DBLP:conf/cav/TouzeauMMR17}, we reimplemented it within \software{Otawa}, as our previous experiments were conducted at the level of the intermediate representation of the LLVM compiler suite rather than on machine code, as our present analysis.
Recall that our analyses and theirs compute exactly the same classifications and differ only in memory and time consumption;
this enabled us to test and debug our implementation.

There are several comparisons that the interested reader would have perhaps appreciated, but that we were unable to perform.
We are not able to directly confront our implementation to \citeauthor{DBLP:conf/cav/TouzeauMMR17}'s because theirs operates, as a proof of concept, upon LLVM's intermediate representation, using a fake memory mapping, while ours operate upon machine code with the true memory mapping.
We were not able to measure the precision gained on the WCET upper bound computed by \software{Otawa} by replacing the imprecise age-based static analysis \cite{DBLP:journals/rts/FerdinandW99} by our precise analysis, due to engineering issues --- our analysis is implemented on top of \software{Otawa} version 2, which is under development and constantly evolving.
We expect to be able to connect our analysis to the WCET computation in \software{Otawa} in a matter of months.
Moreover, our experiments are performed on an instruction cache.
As mentioned in \Autoref{sec:problem}, analyzing data caches is possible but would require further engineering effort to connect to \software{Otawa}'s pointer analyses.

\newcommand{\function}[1]{\software{#1}}

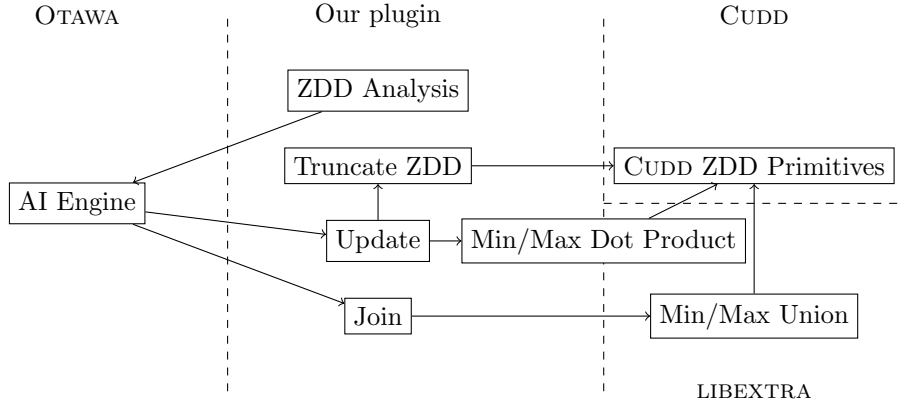
\begin{figure}[h!]
\begin{center}
\begin{tikzpicture}
\draw (2,-1) edge[dashed] (2,4);
\draw (7,-1) edge[dashed] (7,4);
\draw (7,1.5) edge[dashed] (11,1.5);

\node (AI) [draw] at (0, 1.5) {AI Engine};
\node (zdd) [draw] at (4, 3) {ZDD Analysis};
\node (boundsZDD) [draw] at (4, 2) {Truncate ZDD};
\node (update) [draw] at (4, 1) {Update};
\node (join) [draw] at (4, 0) {Join};
\node (primitives) [draw] at (9, 2) {\software{Cudd} ZDD Primitives};
\node (union) [draw] at (9, 0) {Min/Max Union};
\node (dot) [draw,fill=white] at (7, 1) {Min/Max Dot Product};

\draw[->] (update) -- (boundsZDD);
\draw[->] (zdd) -- (AI);
\draw[->] (AI) -- (update);
\draw[->] (AI) -- (join);
\draw[->] (update) -- (dot);
\draw[->] (boundsZDD) -- (primitives);
\draw[->] (join) -- (union);
\draw[->] (union) -- (primitives);
\draw[->] (dot) -- (primitives);

\node () at (0,4) {\software{Otawa}};
\node () at (4,4) {Our plugin};
\node () at (9,4) {\software{Cudd}};
\node () at (9,-1) {\software{libextra}};

\end{tikzpicture}
\end{center}
\caption{Overview of our framework. Edges represent dependencies ($u \rightarrow v$ means that code in $u$ calls some methods in $v$).}
\label{fig:framework}
\end{figure}

\subsection{Implementation and Evaluation of our Analysis}

\Autoref{fig:framework} shows how we integrated our analyses into the \software{Otawa} framework, and how the operations described in \Autoref{sec:zdd} interact with each other.
The main component of our \software{Otawa} plugin is the ``ZDD Analysis'' box, which classifies accesses by abstract interpretation by calling the generic abstract interpretation engine provided by \software{Otawa}.
This iterates on the CFG and calls the update and join operation we provide when needed.
The join operation is realized by computing the union of a given pair of ZDDs, and then removing the subsumed sets (as explained in section \Autoref{sec:zdd}).
This is done by using the \function{MinUnion} (respectively \function{MaxUnion}) function provided by \software{libextra}, which compute the minimal elements of the union of two ZDDs.
The update function models the effect of accessing a block: to this end, the accessed block is added to all sets represented by the current ZDD.
This operation could be performed by the \function{DotProduct} function of \software{libextra} which, given two ZDDs $S_1$ and $S_2$, computes the set $S = \{s_1 \cup s_2, s_1 \in S_1, s_2 \in S_2\}$.
In practice, we use the \function{MaxDotProduct} provided by \software{libextra} which only keeps the maximal elements of $S$.
However, \software{libextra} does not provide the dual \function{MinDotProduct}, whose implementation we added.
Once the new block is added to the current ZDD, we truncate the ZDD, keeping only those sets whose size is below the associativity.
This is achieved by the \function{Truncate} algorithm described in \Autoref{sec:zdd}.

We implemented two different versions of our analysis:
\begin{itemize}
	\item \emph{ZDD}: The version described in \Autoref{sec:zdd}. One analysis is performed for every memory block in the program to classify all accesses to this block.
When an analysis terminates, all the structures it used are freed and \software{Cudd} cache and memoization tables are flushed.
This approach is referred to as \emph{ZDD} in the following.
	\item \emph{Age-based + DU + ZDD}: The last version uses the \emph{Age-based + DU} analyses to classify memory accesses, and refines the accesses left ``unknown'' by running the associated ZDD analysis.
In other words, this approach is the same as \citet{DBLP:conf/cav/TouzeauMMR17}, where the model-checking phase is replaced by our \emph{ZDD} approach.
\end{itemize}

\subsubsection{Comparison of the two Variants of our Analysis}

\begin{figure}
	\begin{subfigure}[h]{0.45\textwidth}
		\centering
		\oldnewFigure{1.0}{\includegraphics[width=\textwidth]{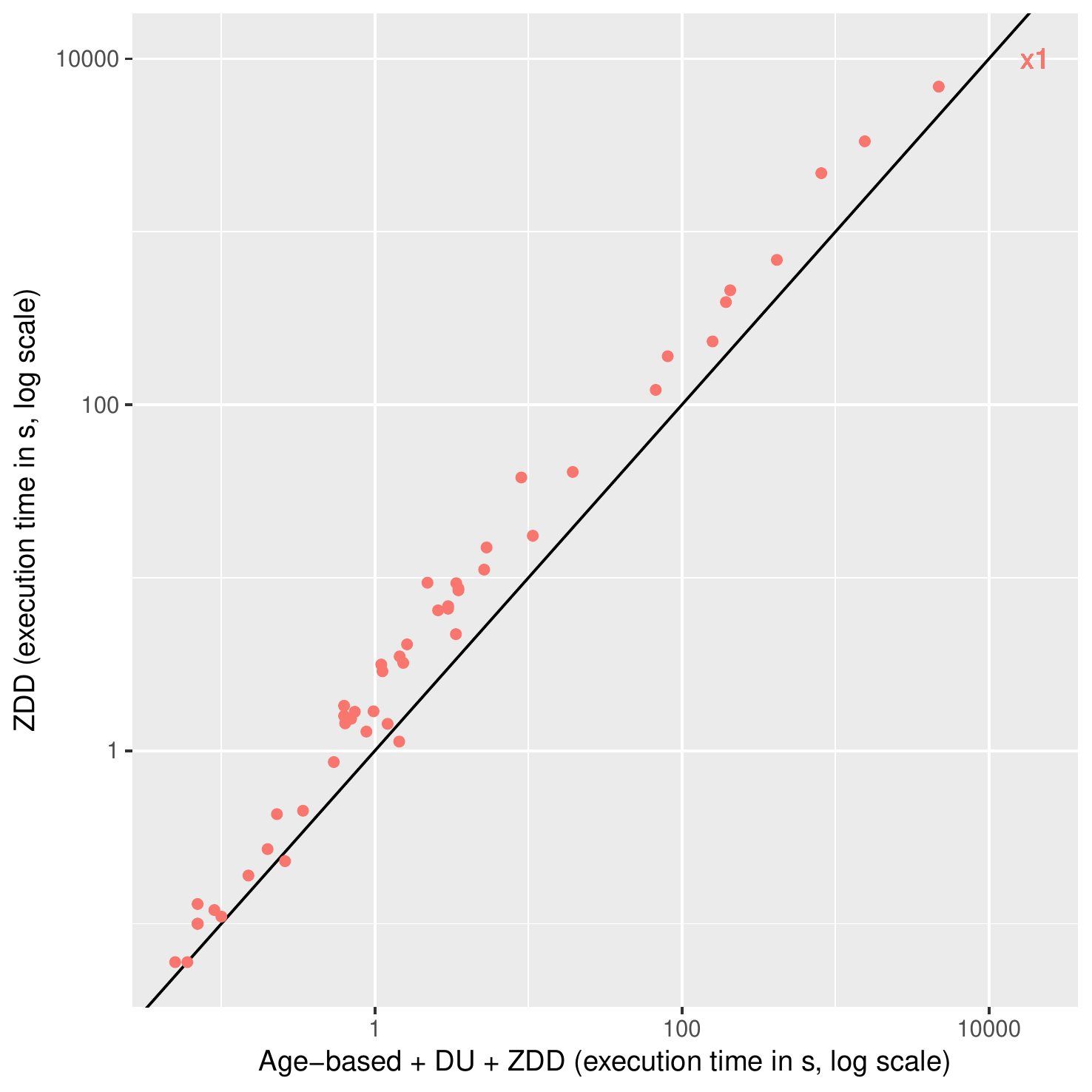}}{\includegraphics[width=\textwidth]{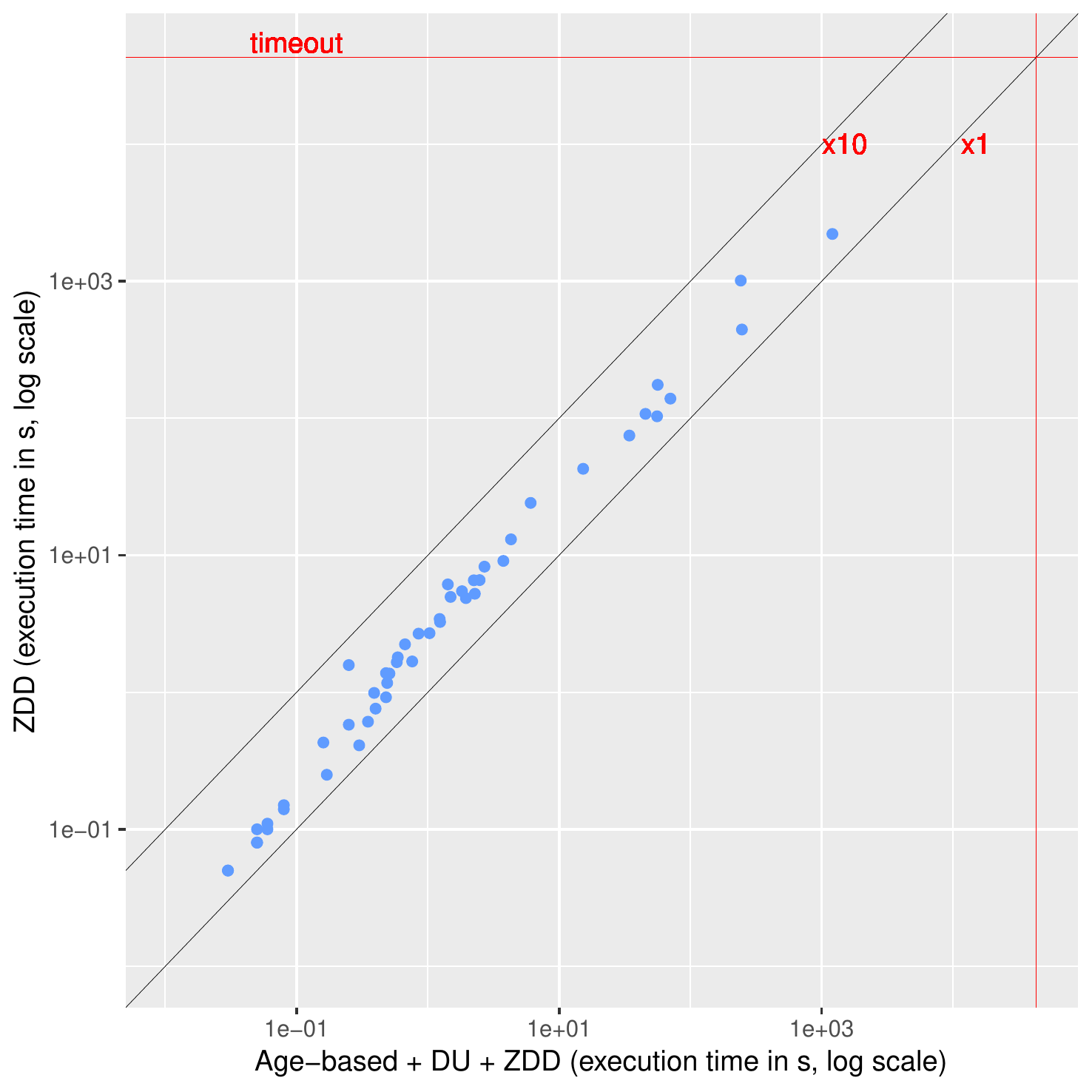}}
		\caption{Time comparison}\label{fig:time_preanalysis}
	\end{subfigure}
	\begin{subfigure}[h]{0.45\textwidth}
		\centering
		\oldnewFigure{1.0}{\includegraphics[width=\textwidth]{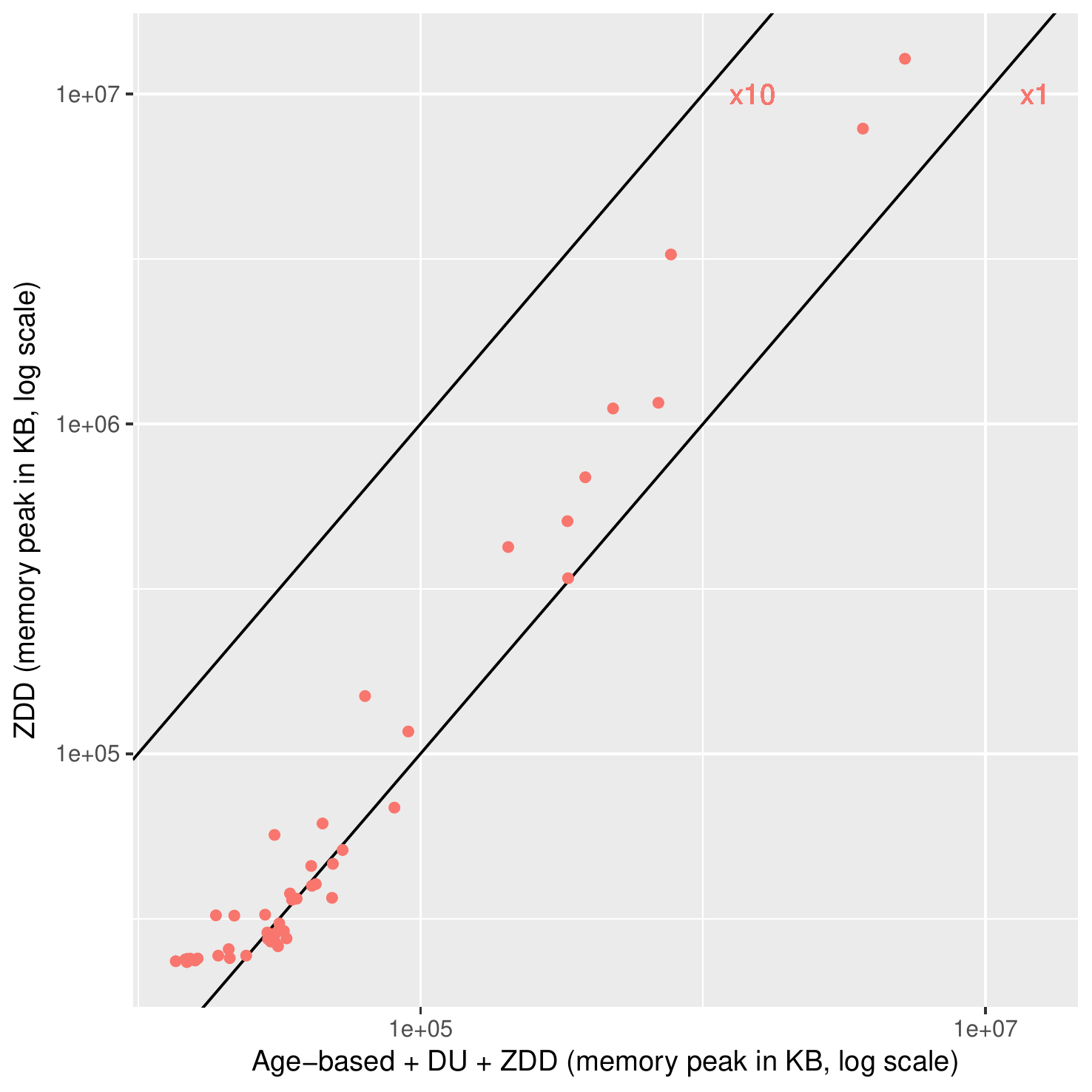}}{\includegraphics[width=\textwidth]{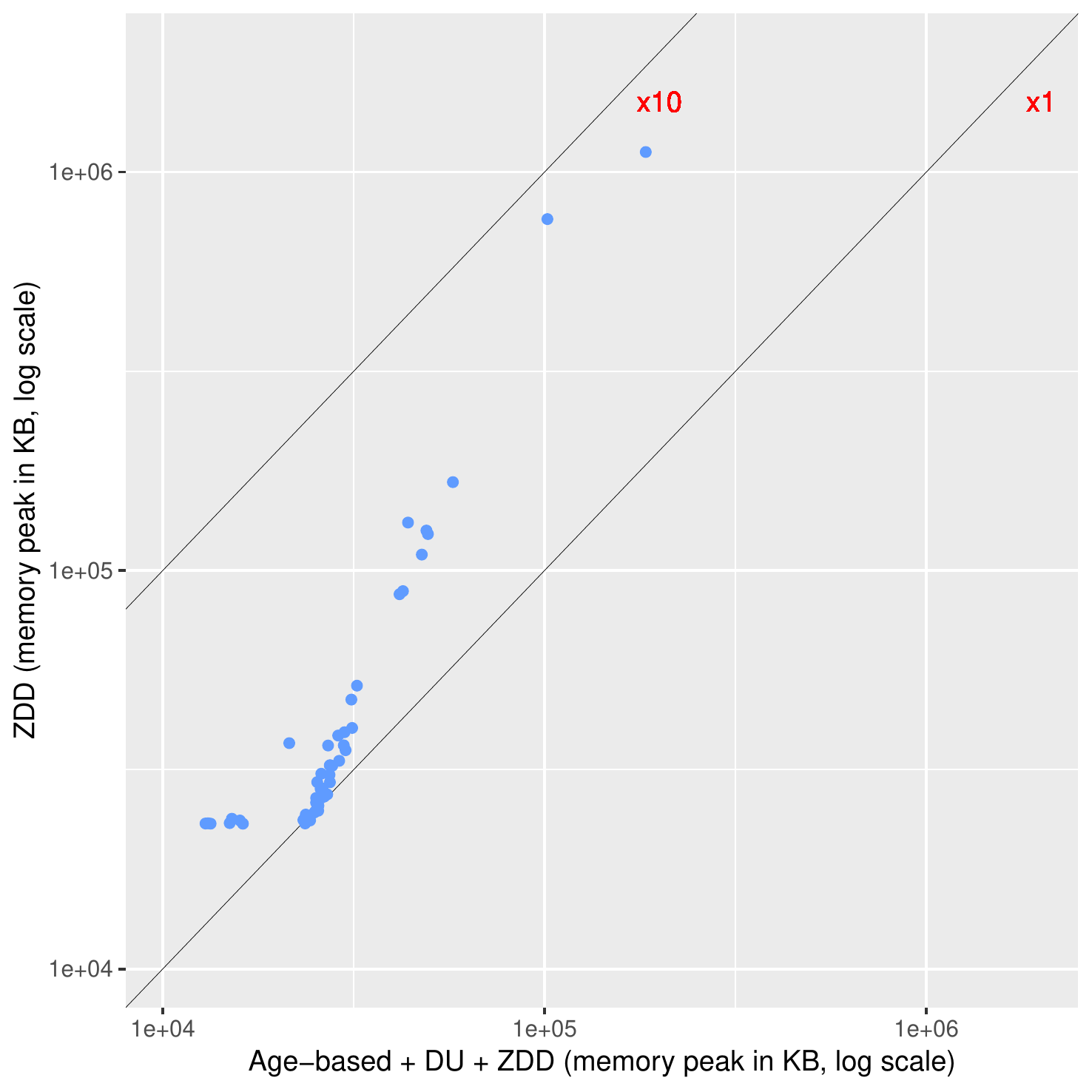}}
		\caption{Memory comparison}\label{fig:memory_preanalysis}
	\end{subfigure}
	\caption{ZDD vs. \emph{Age-based} + DU + ZDD.}
	\label{fig:preanalysis}
\end{figure}

Remember that the two versions of our analysis obtain the same access classifications.
They differ only in the way the classification is obtained.
\Autoref{fig:preanalysis} shows that the approach combining \emph{ZDDs} with a pre-analysis based on the age-based and DU analyses, is more efficient in terms of analysis execution time and memory consumption than the \emph{ZDD} approach alone. 
The pre-analysis is performed using a single pass over the whole program for all memory blocks.
As it successfully classifies most accesses as ``always hit'', ``always miss'', or ``definitely unknown'', the ZDD approach needs to be run only on a relatively small subset of all memory blocks.
We will keep this \emph{Age-based + DU + ZDD} variant as a basis for the following experiments.


\begin{figure}
\begin{center}
\oldnewFigure{0.84}{\includegraphics[width=\textwidth]{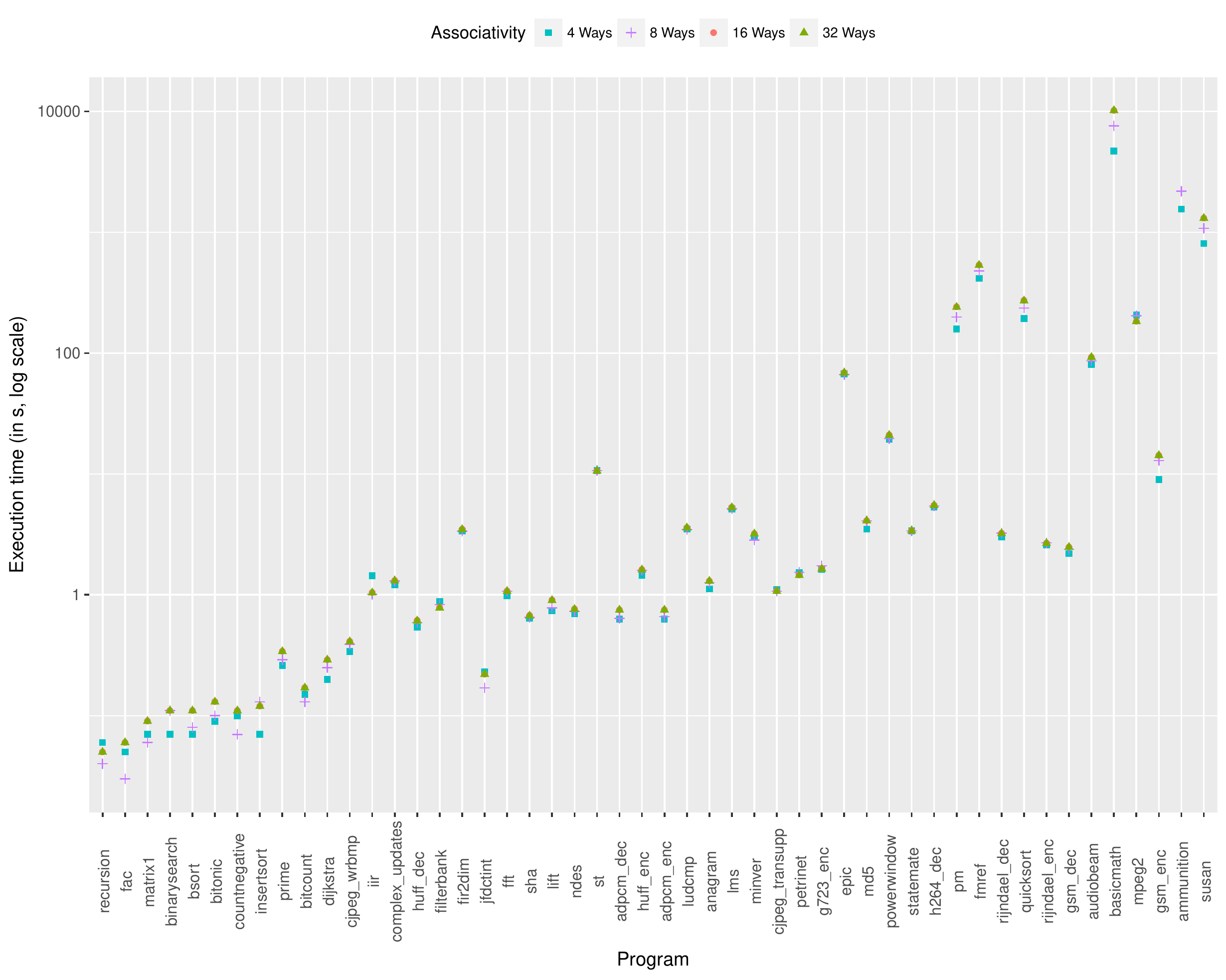}}{\includegraphics[width=\textwidth]{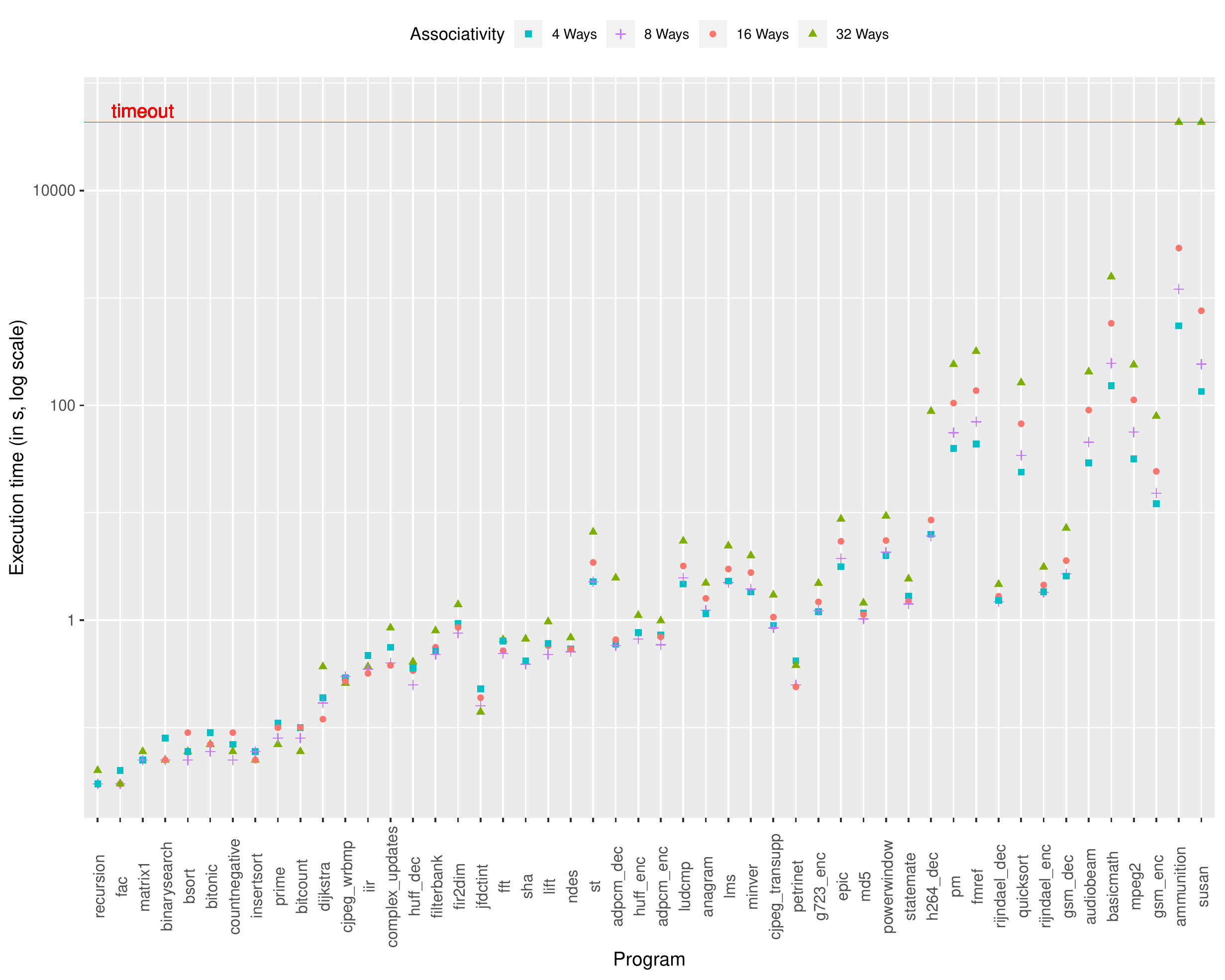}}
\caption{Evolution of execution time of our analysis when increasing associativity.}
\label{fig:associativity}
\end{center}
\end{figure}

\subsubsection{Scalability with Respect to the Associativity}
As mentioned in \Autoref{sec:complexity}, the \emph{may-hit} and \emph{may-miss} problems are NP-complete, when the cache's associativity is considered an input parameter.
In practice, however, repeated accesses to the same block are infrequent and the CFGs' branching structures are simple.
We thus evaluate the running time of our analysis when increasing the associativity of the cache, while keeping the same \oldnew{number of cache sets (thus increasing the cache's capacity)\footnote{Comment to attentive reviewers of this revision: As most benchmarks from the \textsc{TacleBench} are small (see \Autoref{fig:bench_size}), increasing the cache size beyond a certain limit is not sensible: most programs will then fit into the cache entirely. Therefore, to stress the analysis more, we decided to keep the cache size fixed, while increasing the associativity.}.}{cache size (thus decreasing the number of cache sets).}
\Autoref{fig:associativity} shows that our analysis scales well for usual values of associativity;
\oldnew{there is usually only a relatively small increase in analysis times with increasing associativity}{the increase in analysis time is usually proportional to the increase in associativity}.

In the next section we experimentally evaluate how our approach compares with previous work.
Using a first set of experiments, we compare our approaches to the model-checking approach described in \citet{DBLP:conf/cav/TouzeauMMR17}.
Then, we show that our analysis is similarly efficient as the age-based analysis of \citet{DBLP:journals/rts/FerdinandW99} and the DU analysis of \citet{DBLP:conf/cav/TouzeauMMR17} in terms of memory usage and analysis time.

\subsection{Comparison with Prior Work}

Among the existing cache analyses, the approach of \citet{DBLP:conf/cav/TouzeauMMR17} is the closest to our work.
It consists of the following three steps:
\begin{itemize}
	\item First, the usual analysis of \citet{DBLP:journals/rts/FerdinandW99} (which we refer to as \emph{Age-based} analysis in following) is performed and classifies accesses as ``always hit'', ``always miss'' or ``unknown''.
	\item Then, a second analysis, called definitely unknown, classifies a subset of the ``unknown'' accesses as ``definitely unknown'' when it finds both a path leading to hit and a path leading to a miss for a given access (this is an approximated analysis --- it may fail to identify such paths).
	\item Finally, the remaining unknown accesses are classified using a model checker (MC) and marked as ``always hit'', ``always miss'' or ``definitely unknown''.
\end{itemize}
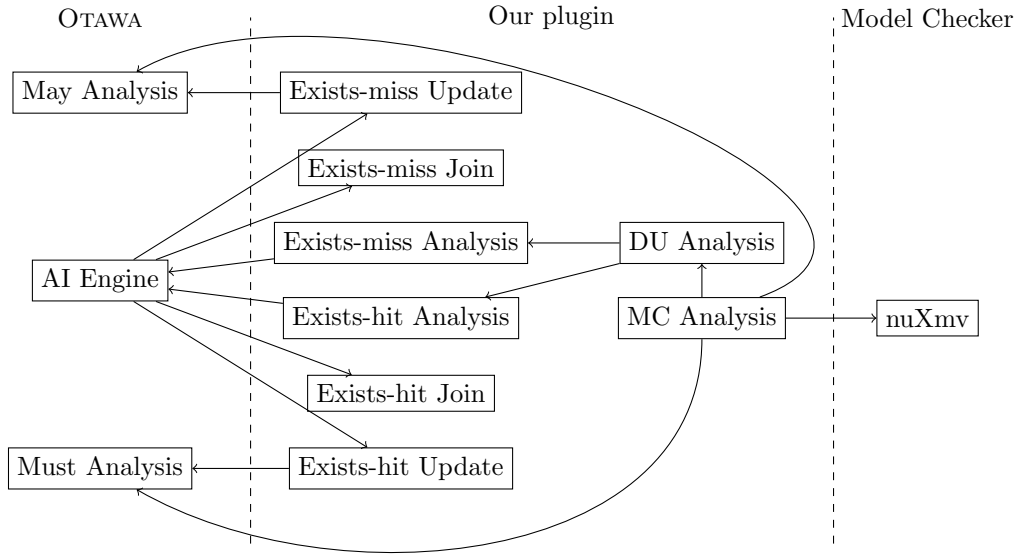
\begin{figure}
\begin{center}
\begin{tikzpicture}
\node (May) [draw] at (0,5) {May Analysis};
\node (AI) [draw] at (0,2.5) {AI Engine};
\node (Must) [draw] at (0,0) {Must Analysis};
\node (EMUpdate) [draw] at (4,5) {Exists-miss Update};
\node (EMJoin) [draw] at (4,4) {Exists-miss Join};
\node (EM) [draw] at (4,3) {Exists-miss Analysis};
\node (EH) [draw] at (4,2) {Exists-hit Analysis};
\node (EHJoin) [draw] at (4,1) {Exists-hit Join};
\node (EHUpdate) [draw] at (4,0) {Exists-hit Update};
\node (DU) [draw] at (8,3) {DU Analysis};
\node (MC) [draw] at (8,2) {MC Analysis};
\node (Nuxmv) [draw] at (11,2) {nuXmv};

\draw[->] (EMUpdate) -- (May);
\draw[->] (EHUpdate) -- (Must);
\draw[->] (AI) -- (EMUpdate);
\draw[->] (AI) -- (EMJoin);
\draw[->] (AI) -- (EHUpdate);
\draw[->] (AI) -- (EHJoin);
\draw[->] (EM) -- (AI);
\draw[->] (EH) -- (AI);
\draw[->] (DU) -- (EH);
\draw[->] (DU) -- (EM);
\draw[->] (MC) -- (DU);
\draw[->] (MC) -- (Nuxmv);
\draw[->] (MC) edge[in=30,out=20] (May);
\draw[->] (MC) edge[in=-30,out=-90] (Must);

\draw (2,-1) edge[dashed] (2,6);
\draw (9.75,-1) edge[dashed] (9.75,6);

\node () at (0,6) {\software{Otawa}};
\node () at (6,6) {Our plugin};
\node () at (11,6) {Model Checker};

\end{tikzpicture}
\end{center}
\caption{Integration of \citet{DBLP:conf/cav/TouzeauMMR17} in \software{Otawa}.}
\label{fig:implem_CAV}
\end{figure}

Our implementation of this approach in \software{Otawa} is illustrated in \Autoref{fig:implem_CAV}.
Note that the DU analysis is based on two approximate analyses: an "exists hit'' analysis, which determines whether a path leading to a hit exists, and an ``exists miss'' analysis which determines whether a path leading to a miss exists.
Note that the DU analysis reuses the abstract cache states from \emph{Age-based} May/Must analyses: in our implementation, \emph{Age-based} analysis is provided by \software{Otawa}.
Finally, the accesses that are not classified precisely by the two abstract interpretation phases are refined using a call to the \software{nuXmv} model checker, which processes a model focussed on the memory block under analysis.

\subsubsection{Comparison to \citet{DBLP:conf/cav/TouzeauMMR17}}

\begin{figure}
	\begin{subfigure}[h]{0.45\textwidth}
		\centering
		\oldnewFigure{1}{\includegraphics[width=\textwidth]{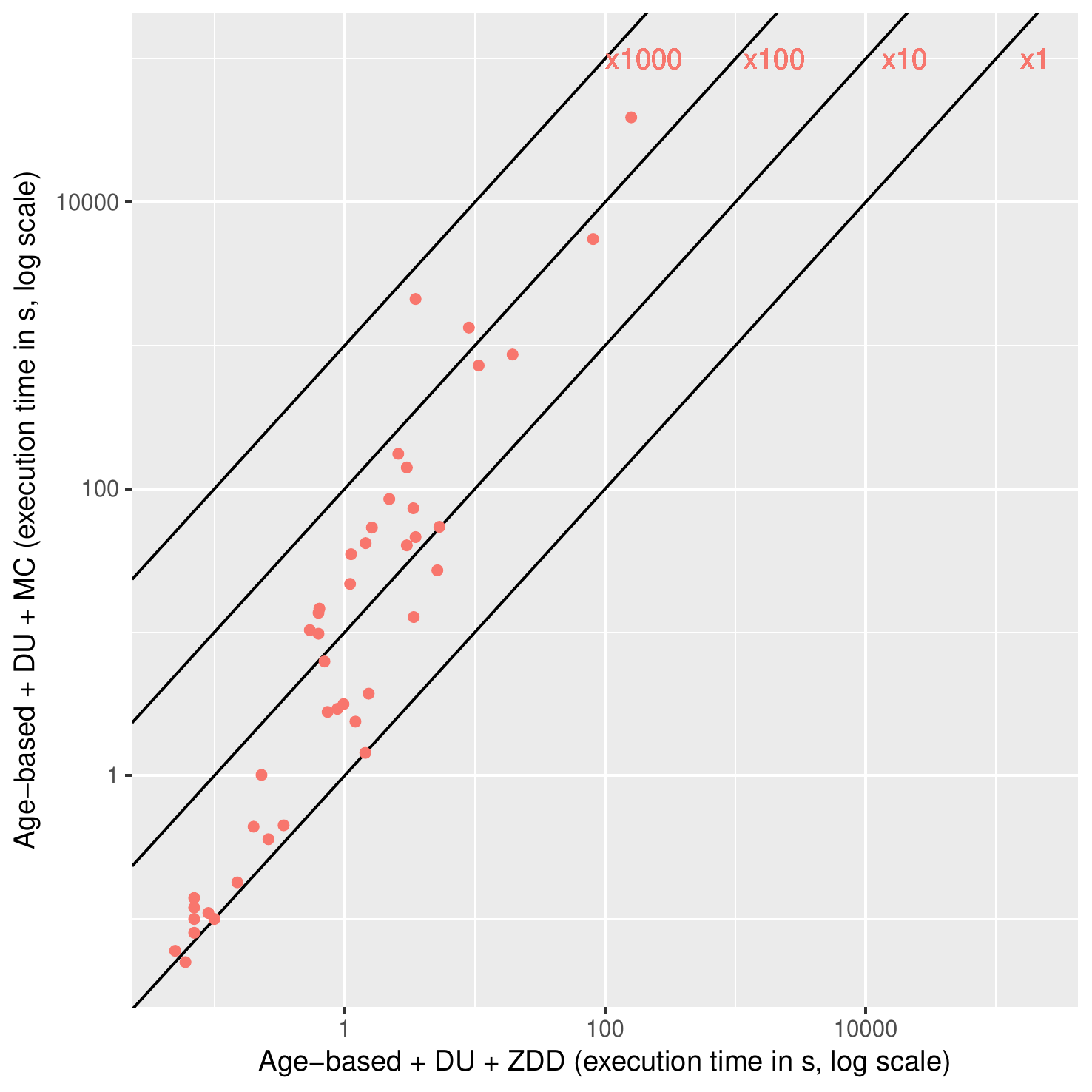}}{\includegraphics[width=\textwidth]{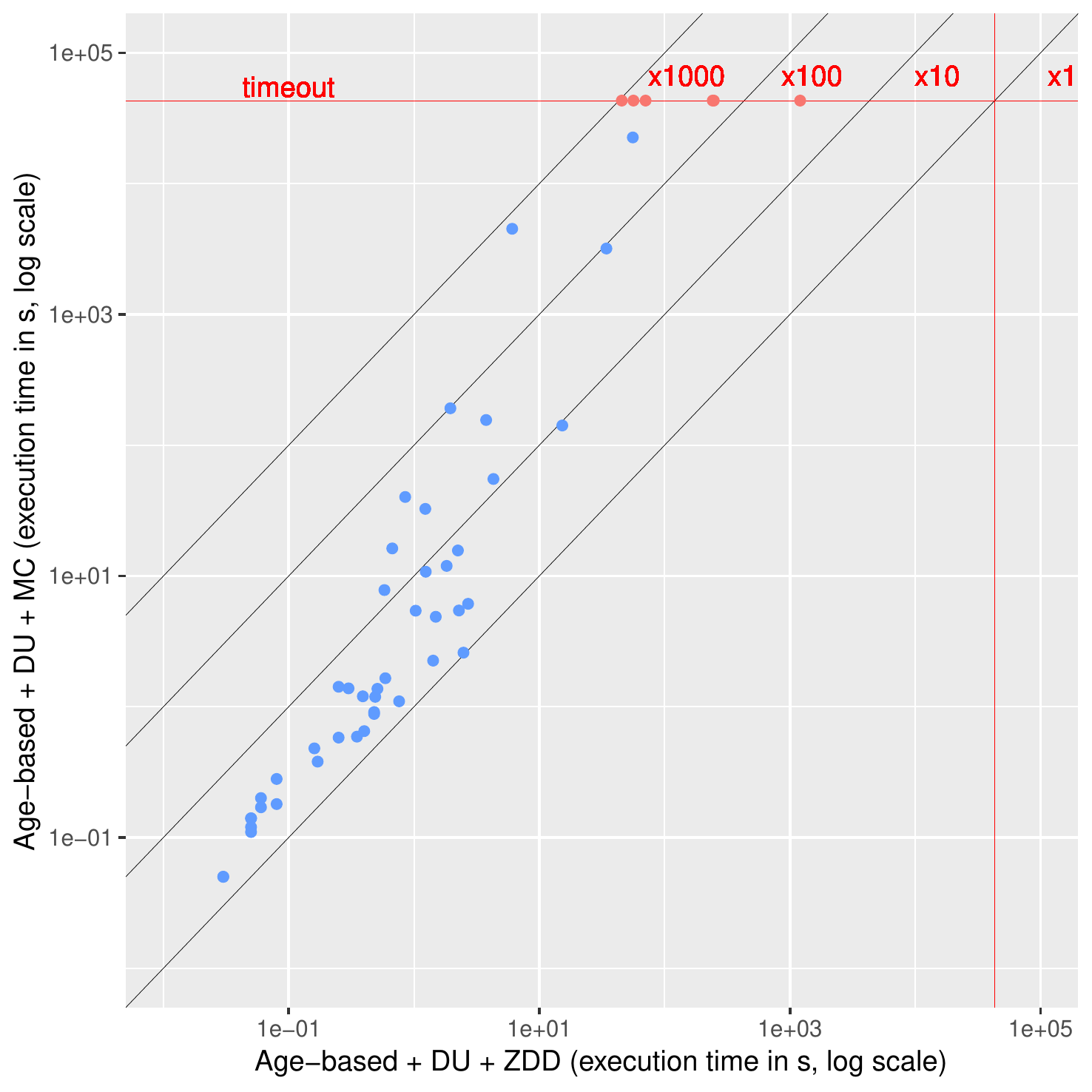}}
		\caption{Time comparison}\label{fig:time_vs_cav}
	\end{subfigure}
	\begin{subfigure}[h]{0.45\textwidth}
		\centering
		\oldnewFigure{1}{\includegraphics[width=\textwidth]{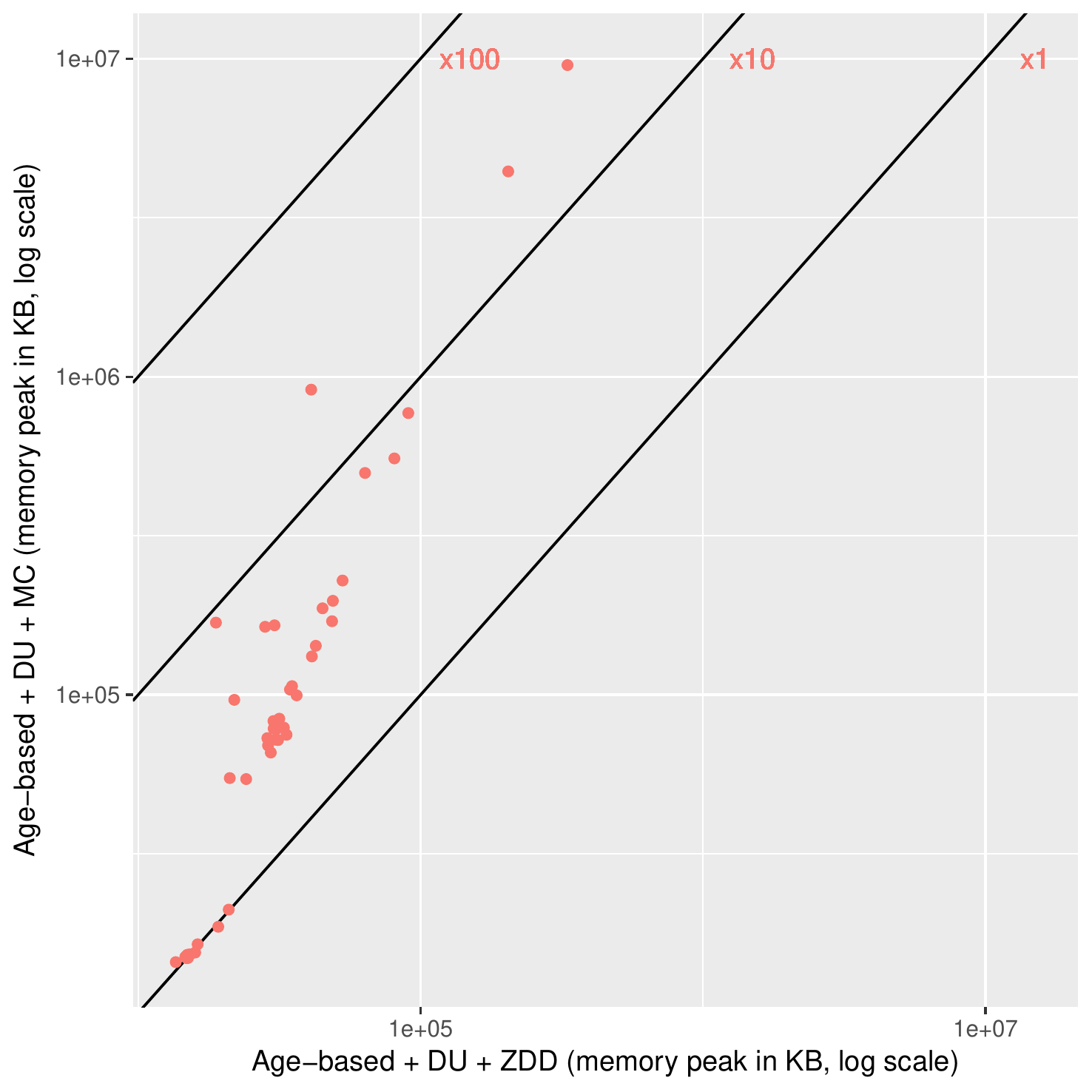}}{\includegraphics[width=\textwidth]{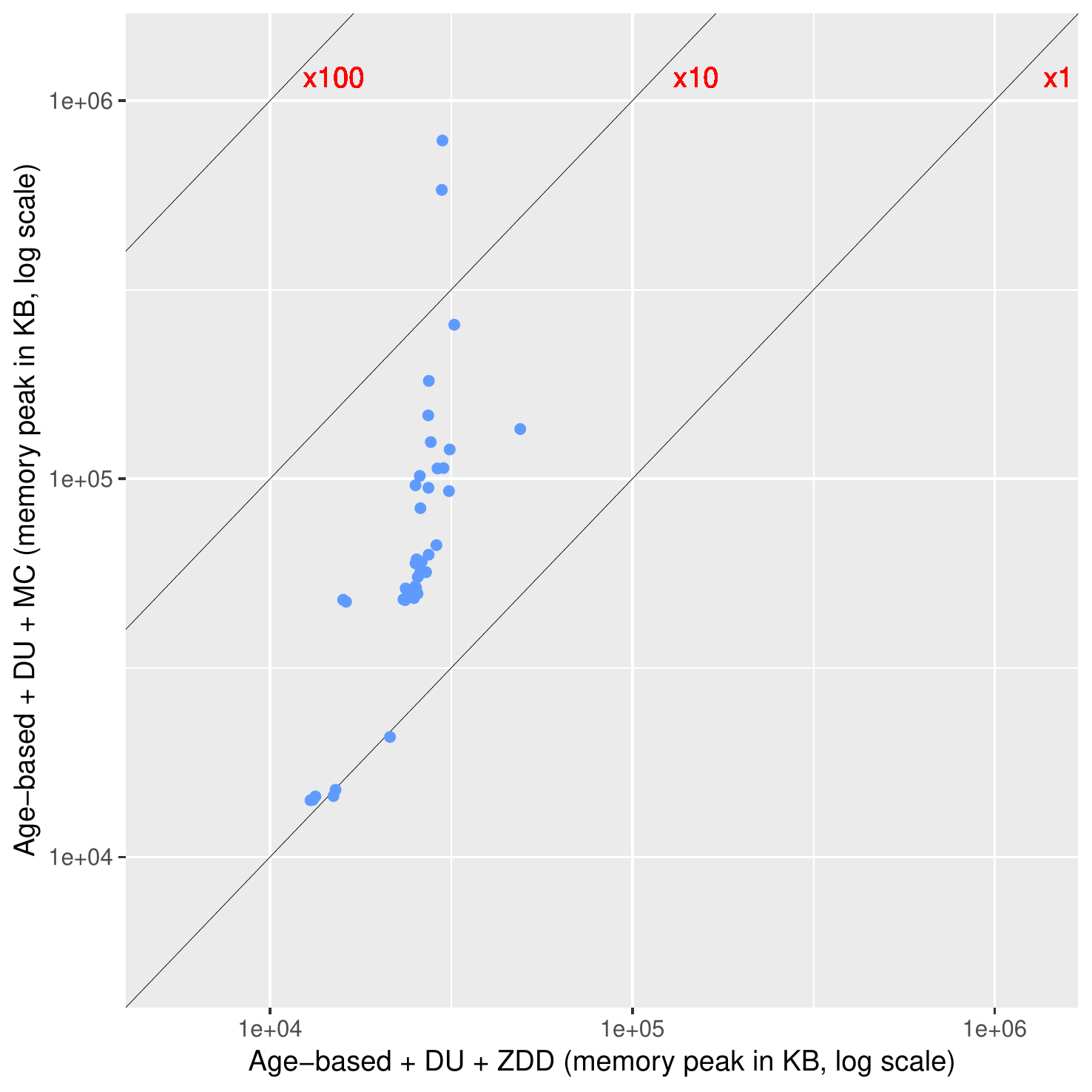}}
		\caption{Memory comparison}\label{fig:memory_vs_cav}
	\end{subfigure}
	\caption{ZDD approach vs. model-checking approach from \cite{DBLP:conf/cav/TouzeauMMR17}.}
	\label{fig:vs_cav}
\end{figure}

As the prior work of Touzeau et al. and ours provide the same classification of memory accesses, we compare the analyses according to two criteria:
\begin{enumerate}[a)]
	\item The full analysis time, including every analysis step from CFG reconstruction to memory-access classification, is compared in \Autoref{fig:time_vs_cav}.
	\item The peak memory usage of the two approaches is compared in \Autoref{fig:memory_vs_cav}.
\end{enumerate}
The scatter plots are both on a log. scale.
Each dot corresponds to the resource consumption of one benchmark from \software{TacleBench} under the two analysis approaches.

The figures clearly show that the ZDD approach is significantly faster than the model-checking approach (more than a hundred times faster for the largest benchmarks) and that the benefits increase with the size of the benchmarks.
The peak memory usage is also generally smaller with our ZDD approach than with model checking.


\subsubsection{Overhead of ZDD over \emph{Age-based + DU}}

\begin{figure}
\begin{center}
\oldnewFigure{0.9}{\includegraphics[width=0.85\textwidth]{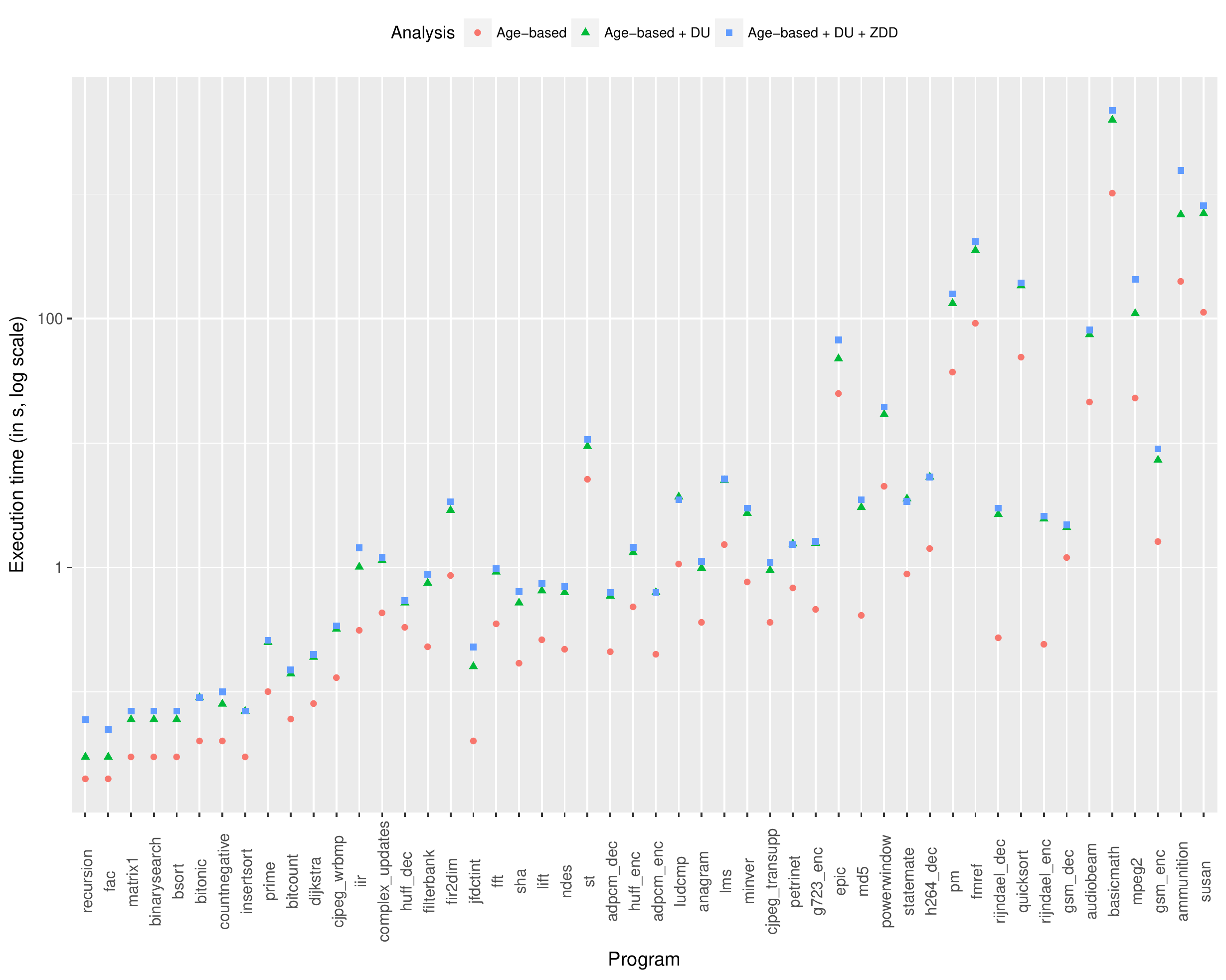}}{\includegraphics[width=0.85\textwidth]{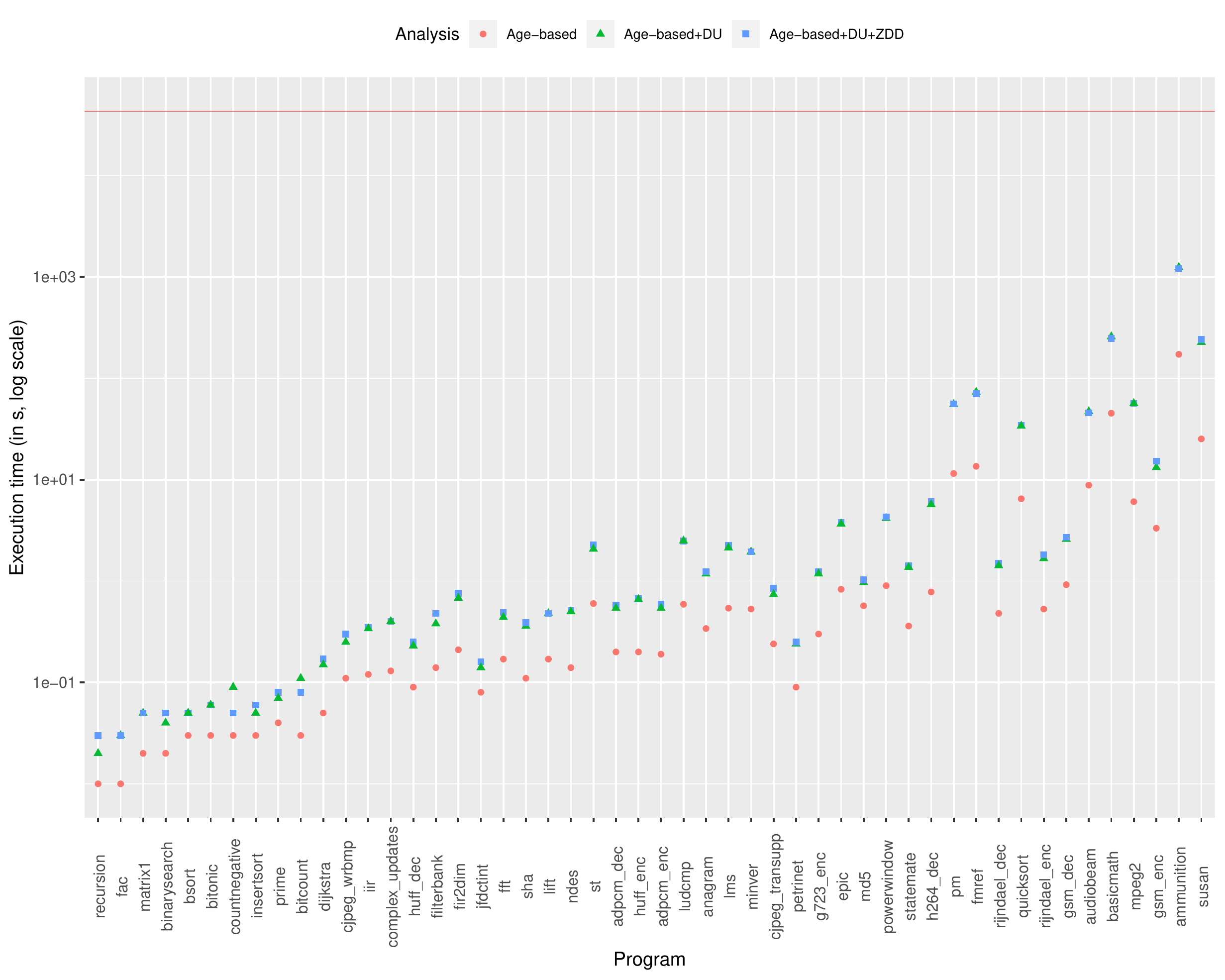}}
\caption{Execution time overhead.}\label{fig:time_overhead}
\end{center}
\end{figure}
\begin{figure}
\begin{center}
\oldnewFigure{0.9}{\includegraphics[width=0.85\textwidth]{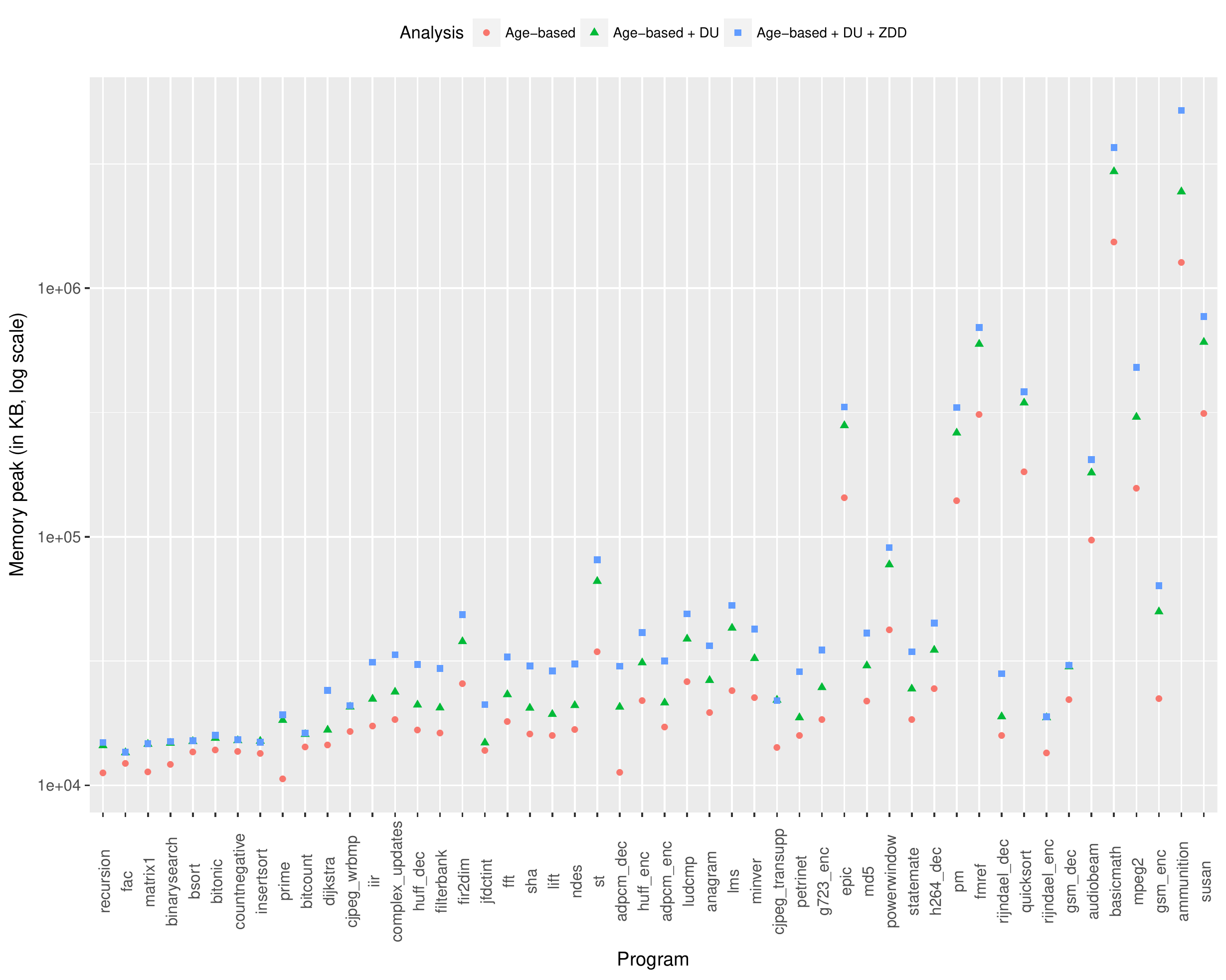}}{\includegraphics[width=0.85\textwidth]{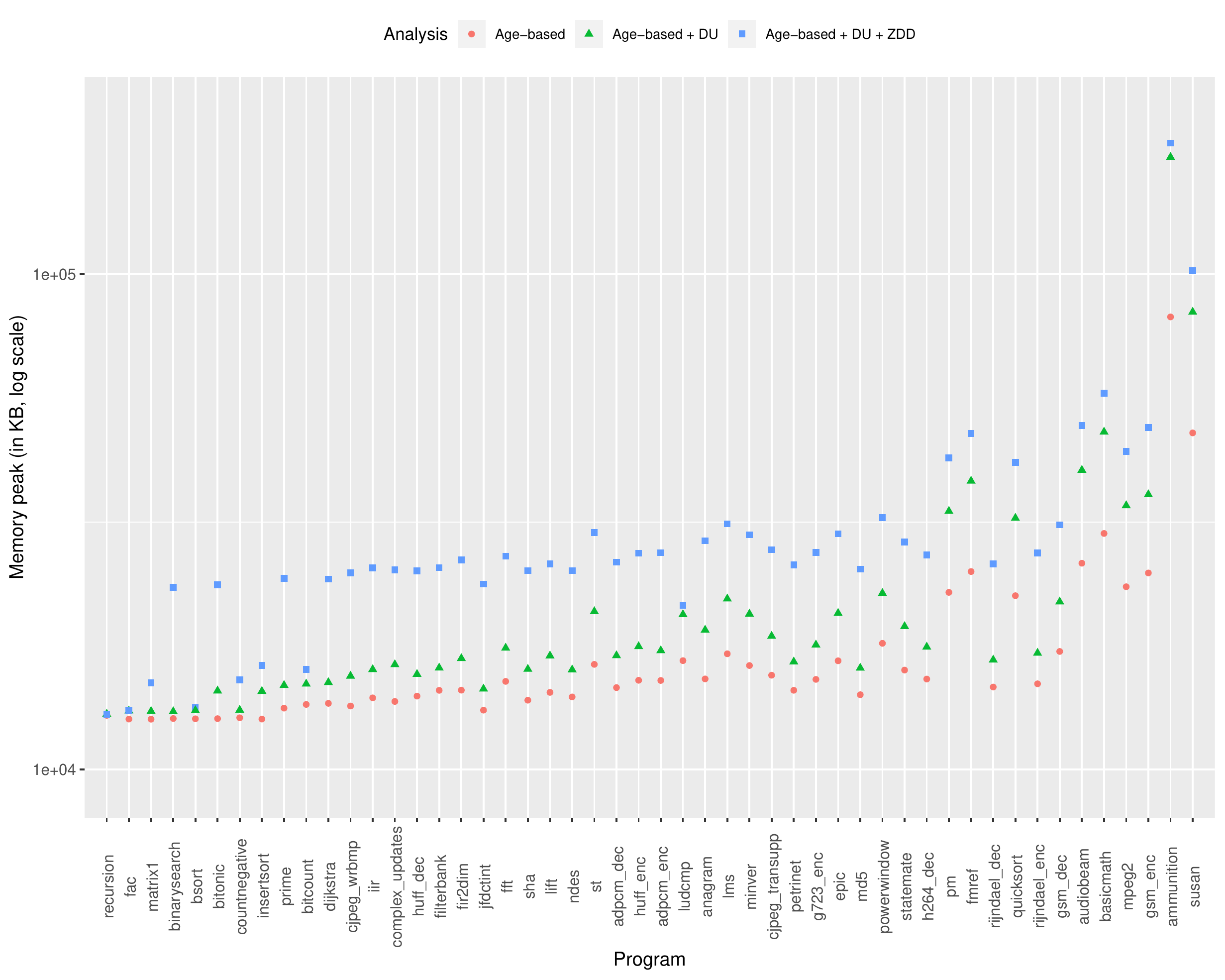}}
\caption{Memory overhead.}\label{fig:memory_overhead}
\end{center}
\end{figure}

We evaluate the scalability of our approach comparing it to the usual \emph{Age-based} analysis and DU analysis.
\Autoref{fig:time_overhead} and~\Autoref{fig:memory_overhead} show the execution time of our  approach compared to \emph{Age-based} and \emph{Age-based + DU} on a logarithmic scale.
The vertical distance between two points can be interpreted as the slowdown due to the additional analysis\footnote{Note that for very small benchmarks the time measurements are not very reliable, as the reported values correspond to single measurements. Experiments with \textsc{countnegative} showed variations of up to 30\% from one measurement to another.}.
On average, our approach is \oldnew{3.53}{3.46} times more costly than the usual \emph{Age-based} analysis, and only adds a \oldnew{16.4\%}{4.6\%} overhead over the \emph{Age-based + DU} analysis, as shown in \Autoref{fig:scatter_overhead}.

\begin{figure}
\begin{center}
\oldnewFigure{1.0}{\includegraphics[width=.5\textwidth]{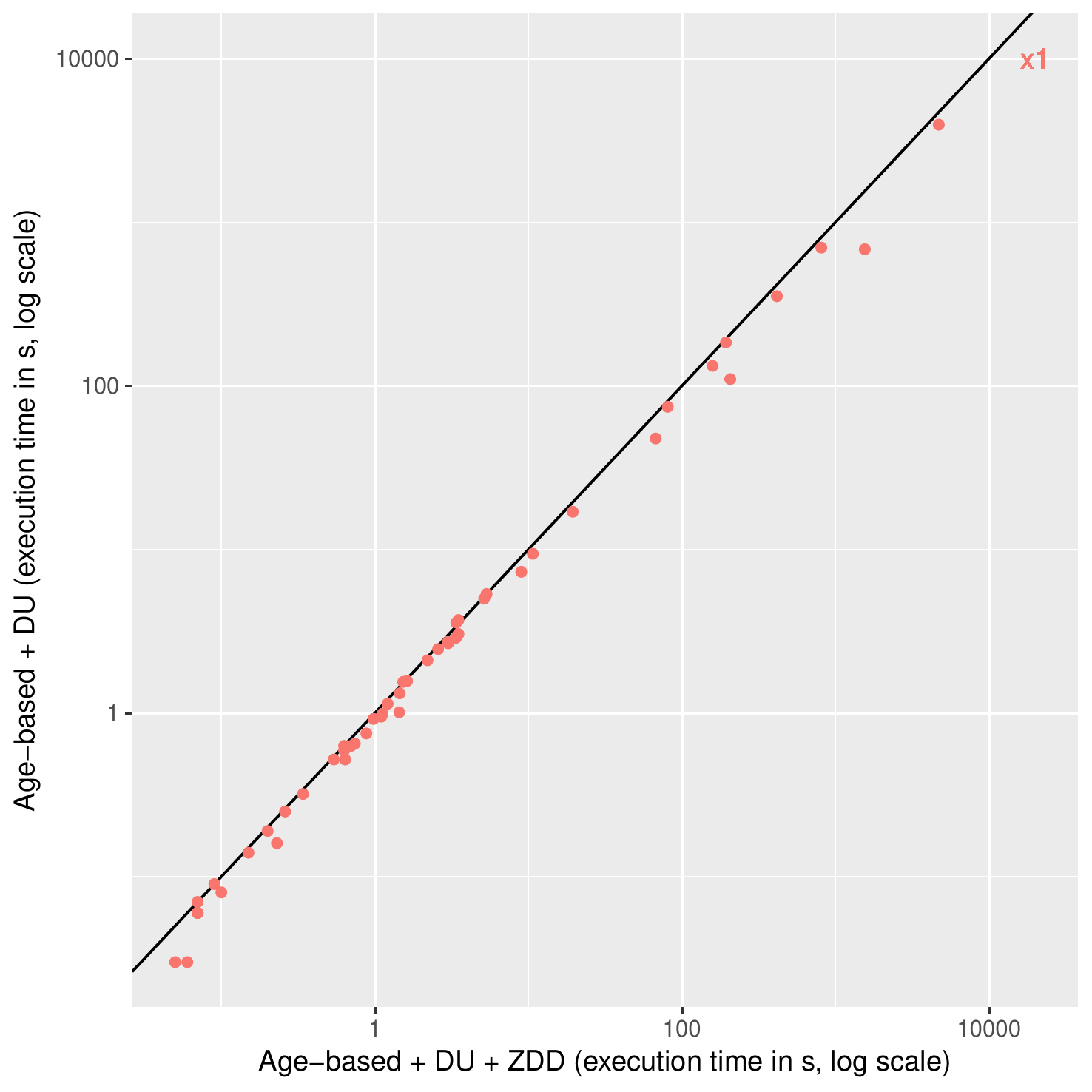}}{\includegraphics[width=.5\textwidth]{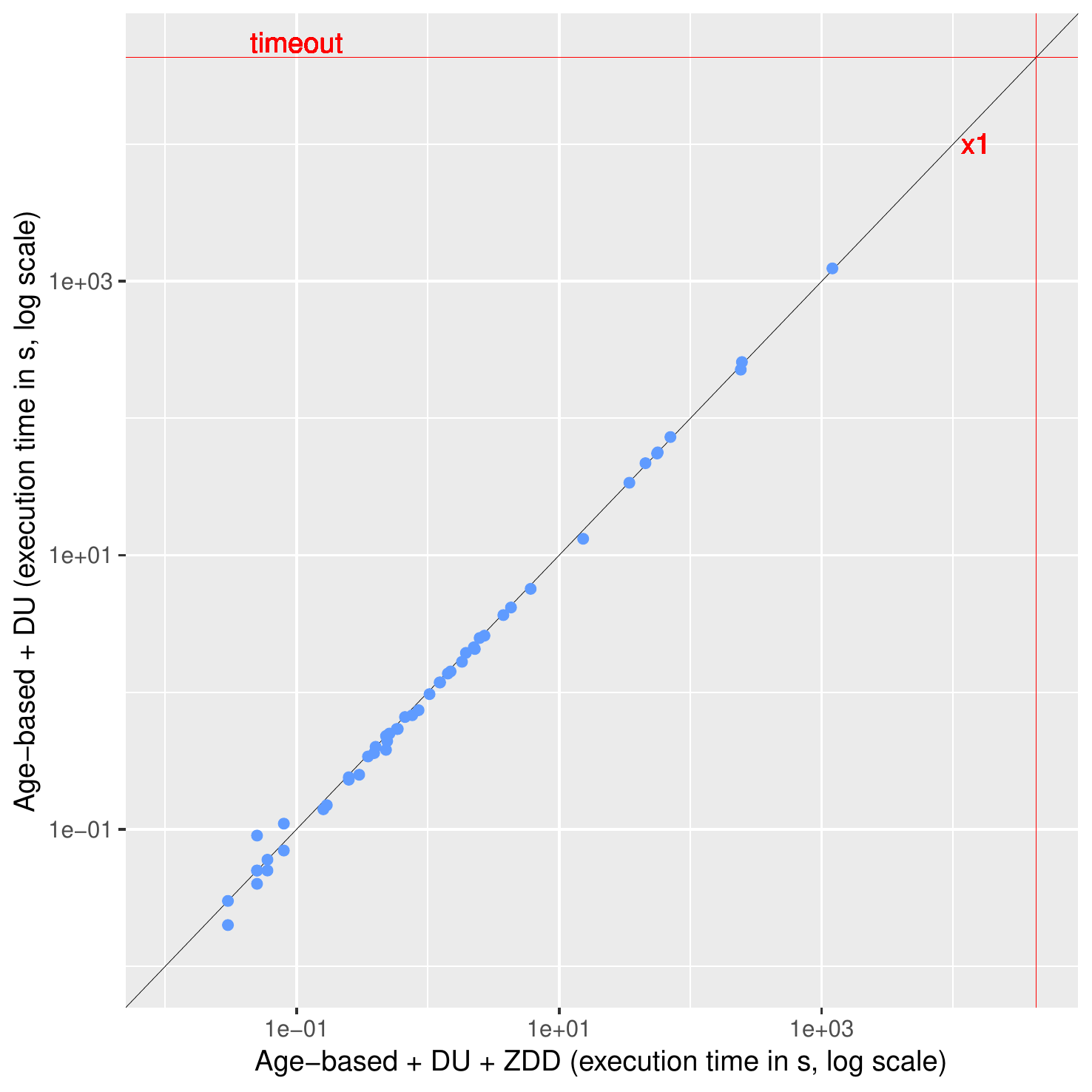}}
\caption{Execution time overhead.}\label{fig:scatter_overhead}
\end{center}
\end{figure}

To conclude, the ZDD approach is significantly faster than the exact analysis by \citet{DBLP:conf/cav/TouzeauMMR17}.
This validates our choice of algorithms and data structures. 
The results demonstrate that the ZDD approach offers good scalability that could lead to an industrial use.

\subsubsection{Scalability Comparison with Other Previous Work}
\label{sec:scalability_other_related_work}
The approaches of \citet{Cha_rts13,Chu_rtas16} analyze programs at the source code level and are thus not directly comparable to our work (which analyzes binary code), if only because we map statements to cache blocks differently.

The only common benchmark with ours is \software{statemate}: for this benchmark their analysis stops after 100 calls to the model checker and the analysis spends 195 seconds, where we analyze the whole benchmark in less than \oldnew{4}{2} seconds.
 Similarly, in \citet{Chu_rtas16} the paper states an analysis time of 350 seconds for \software{statemate} and 38 seconds for the benchmark \software{ndes} where we analyze them in less than \oldnew{4}{2} seconds and less than \oldnew{1}{1} second.

Unfortunately, there are no other common benchmarks; it however seems our analysis scales much better than these two previous works.


%% file: relatedwork.tex
\section{Related Work}
\label{sec:related_work}
\paragraph{Cache analysis for the verification of real-time systems}
Static cache analysis was first studied in the context of real-time systems~\cite{LITESLITES-v003-i001-a005}.
\citet{Mueller1995ss} introduced a data-flow analysis for direct-mapped caches, i.e., for caches with associativity 1.
Based on abstract interpretation (AI), \citet{DBLP:journals/rts/FerdinandW99} proposed the classical age-based analysis for set-associative caches with LRU replacement.
Their analysis is still in widespread use in commercial and academic WCET analysis tools, e.g.~\cite{ait,DBLP:conf/seus/BallabrigaCRS10}.

As discussed in \Autoref{sec:motivating_example}, Ferdinand and Wilhelm's analysis can be seen as computing a range of possible ages for each memory block in a given program.
On straight-line code this analysis is exact.
However, at control-flow joins the relation between ages of different blocks may be lost. 
As a consequence, the analysis may classify some accesses as ``unknown'' that are in fact always hits or always misses.
There have been several attempts to improve upon the precision of the classical AI-based cache analysis:

\citet{Cha_rts13} refines memory accesses classified as ``unknown'' by AI using a software model-checking step: when abstract interpretation cannot classify an access, the source program is enriched with annotations for counting conflicting accesses and run through a software model checker.
		Their approach, in contrast to ours, takes into account program semantics during the refinement step;
it is thus likely to be more precise on programs where many paths are infeasible for semantic reasons.
Our approach however scales considerably better, as shown in \Autoref{sec:implementation}.
(Our approach can also be combined with analyses for checking the feasibility of paths, as sketched in \Autoref{ref:future_work}; we however have not experimented with it yet).
They advocated applying their method with a time bound. Accesses classified as ``unknown'' by AI are refined until the time bound is reached.
\citet{Chu_rtas16} present a WCET analysis framework based on symbolic execution, where an SMT solver is used to prune infeasible paths.
They employ the age-based abstraction of Ferdinand and Wilhelm within symbolic execution, but never join states, thus avoiding any imprecision.
See Section~\ref{sec:scalability_other_related_work} for a performance comparison of these two analyses with ours.

\citet{DBLP:conf/cav/TouzeauMMR17} refine accesses classified as ``unknown'' by AI using model checking similarly to \citet{Cha_rts13}.
In order to reduce the number of calls to the model checker, they introduce an AI-based analysis that can classify accesses as ``definitely-unknown''.
For a ``definitely-unknown'' access both a hit and a miss are possible depending on the path taken through the control flow graph to reach the access. The classification of such accesses cannot be refined by model checking.
To reduce the effort of the model checker in classifying a particular access, they also introduce a ``focused semantics'', which we reuse and extend in this paper, as discussed in \Autoref{sec:fixpoint}.
We compare the efficiency of their approach with ours in \Autoref{sec:implementation}.

Abstract interpretation has also been applied to the analysis of caches with other popular replacement policies found in modern microarchitectures, such as \emph{first-in, first-out}~(FIFO)~\cite{Grund09,Grund10,Guan13}, \emph{not most-recently-used}~(NMRU)~\cite{Guan14}, and \emph{pseudo-LRU}~(PLRU)~\cite{Grund10a}.
No exact analysis has been proposed for these policies, which are considered to be harder to analyze than LRU~\cite{Reineke07}.
It is doubtful whether our approach can be extended to these policies, as they do not seem to exhibit any useful monotonicity properties as LRU does.

\paragraph{Compiler optimizations}
Optimizing compilers may apply loop transformations to maximize parallelism and data locality~\cite{Feautrier1992,Lim:1997:MPM:263699.263719}.
To support such optimizations, various approaches have been proposed to compute or approximate the number of cache misses of a given loop nest~\cite{Ghosh:1999:CME:325478.325479,Chatterjee:2001:EAC:378795.378859,Cascaval:2003:ECM:782814.782836,BEYLS2005223,Bao:2017:AMC:3177123.3158120}.
This line of work is limited to restricted classes of programs, usually affine loop nests with no input-dependent control flow or input-dependent memory accesses.
The advantage of such methods is that they distinguish each dynamic instance of an instruction in a loop nest, while our approach classifies all dynamic instances together, thereby introducing pessimism.
It would be interesting to investigate whether the exact abstraction developed in this paper could be combined with analytical approaches such as \cite{Bao:2017:AMC:3177123.3158120} to support input-dependent program behavior. 

\paragraph{Cache side-channel analysis}
Caches can be exploited as covert channels~\cite{Kocher2018,Lipp2018} and in side-channel attacks~\cite{Bernstein_2005,Mowery:2012:AXC:2381913.2381917,Liu2015,Yarom2017}.
Static cache analysis has been applied to quantify the vulnerability of implementations of cryptographic protocols~\cite{Doychev2015, Doychev:2017:RAS:3062341.3062388} to cache side-channel attacks.
More accurate cache analyses, such as the one developed in this paper, may yield more accurate vulnerability quantifications.
Applying our analysis in this context is future work.

\paragraph{Antichains}
Antichains have been used in verification to represent lower and upper sets, which occur in many contexts (automata, LTL satisfiability, games, etc.).
\citet{DBLP:conf/cav/WulfDHR06,DBLP:conf/atva/WulfDMR08} proposed two succinct representations:
\begin{inparaenum}[(i)]
\item \emph{fully symbolic}: a binary decision diagram (BDD) represents sets of sets of states: each automaton state is mapped to a BDD variable, a set of states is thus a valuation;
  this representation is thus similar to ours except that they use BDDs and not ZDDs;
\item \emph{semi symbolic}: a state is encoded as an integer, thus as a vector of bits, a set of states is thus encoded as a BDD, and an antichain is thus a sequence of BDDs;
  this representation is thus similar to our explicit list of sets of blocks, which we tried then discarded due to inefficiency.
\end{inparaenum}

They report that, on their examples, the fully symbolic representation is less efficient than the semi symbolic one, particularly on large automata;
the explain this by the linear growth of the number of variables in the BDDs for the fully symbolic representation with respect to automaton size,
as opposed to logarithmic growth in the semi symbolic representation.
We explain this difference with our own findings as follows:
\begin{inparaenum}[(1)]
\item Their sets of states correspond to sets of blocks in our problem. We limit the cardinality of the sets we handle to associativity, whereas, as far as we understand, their sets can be very large.
\item We use ZDDs, which are more compact than BDDs for sets of small sets (no need for nodes ``this element is not in the set'').
\item Their antichains are smaller than ours. 
\end{inparaenum}
Our ZDD approach thus seems efficient for large, often similar antichains of small sets, whereas their semi symbolic approach seems efficient for small antichains of large sets.


%% file: conclusion.tex
\section{Conclusion and Future Work}
\label{sec:conclusion}
\label{ref:future_work}

For decades, it was believed that only rough abstractions (Ferdinand's analysis) could scale up for cache analyses. We show here that it is actually possible to perform an exact analysis by carefully refining the abstraction and using good algorithms and data structures.

We have demonstrated how it is possible to obtain a cache analysis as precise as one obtained by analyzing the cache replacement policy using a model checker, but at a much lower cost (hundreds of times faster).
The main difference between the two analysis is that we abstract the problem further while preserving the exact same results.

Our results are the strongest possible (exactly classifying accesses as ``always hit'', ``always miss'', ``hits or misses depending on the execution'') with respect to a model where all paths inside the control-flow graph may be taken.
This includes paths that cannot be taken inside the program, e.g. ones with conflicting tests $x < 1$ and $x > 2$ with no change to $x$ in between.
It is impossible to remove all spurious paths: this is an undecidable question.
We can however combine our analysis with others to improve the precision in this respect.

Our analysis computes an abstract state, as opposed to encoding everything in a model checker. As such, it can be combined simultaneously with other abstract interpretations (program variables, pointers, micro-architecture\dots) and approaches such as the \emph{abstract reachability graph}, with nodes adorned by pairs (location, abstract state), edges adorned by instructions and where each test, at least up to a certain depth, is split into two branches;
two nodes may be merged, as in a directed acyclic graph, if one abstract state is subsumed by the other;
cycles in the graph may be introduced, perhaps after widening operations, to account for loops.
Such an analysis would exclude some infeasible paths.

Our approach applies both to instruction and data caches. If the address of a data fetch or write is only known to lie within a set of addresses $\{a_1, \dots, a_n\}$, then we consider n parallel edges labeled with $a_1, \dots, a_n$ in the analysis graph. As future work, trace partitioning may be used to improve the precision of such an approach, and can be implemented as a layer above our analysis.


%% file: sharing.tex
\section{Sharing ZDD implementation}

BDD (or ZDD) libraries intensively use hash tables for
\begin{inparadesc}
\item[hash-consing] the BDD nodes: when the library wishes to create a node isomorphic to one that already exists in the system, that one is used instead;
\item[memoizing] the BDD operations: during a recursive operation $f$ on $n$-tuples of BDDs nodes, the computed values of $f$ are stored into a hash table and are retrieved if the same $n$-tuple is encountered instead of recursing;
this ensures that an operation of BDDs $D_1,\dots,D_n$ is (roughly) in time $\prod_i |D_i|$.
\end{inparadesc}
The requirements are that all currently reachable BDD nodes should be retained in the hash-consing table (but other nodes may be retained), and that (at least for ensuring polynomial computations) during one BDD operation the memoizing hash tables should not be flushed (but they may be retained longer).
Depending on the BDD library in use and its parameters, ``garbage'' in these tables may be collected eagerly (at the risk of having to recreate or recompute collected data) or late (at the risk of storing useless data).

In our initial implementation, the analyses for various focus blocks $a$ were run completely separately; all ZDD tables were flushed in between.
We implemented a variation of the \emph{ZDD} approach that does \emph{not} free all structures and \software{Cudd}'s data between two analysis runs for different memory blocks.
Doing so, the analysis may benefit from computations done in a previous analysis run if the result of a particular computation has been memoized.

\begin{figure}[h]
	\begin{subfigure}[h]{0.45\textwidth}
		\centering
		\includegraphics[width=\textwidth]{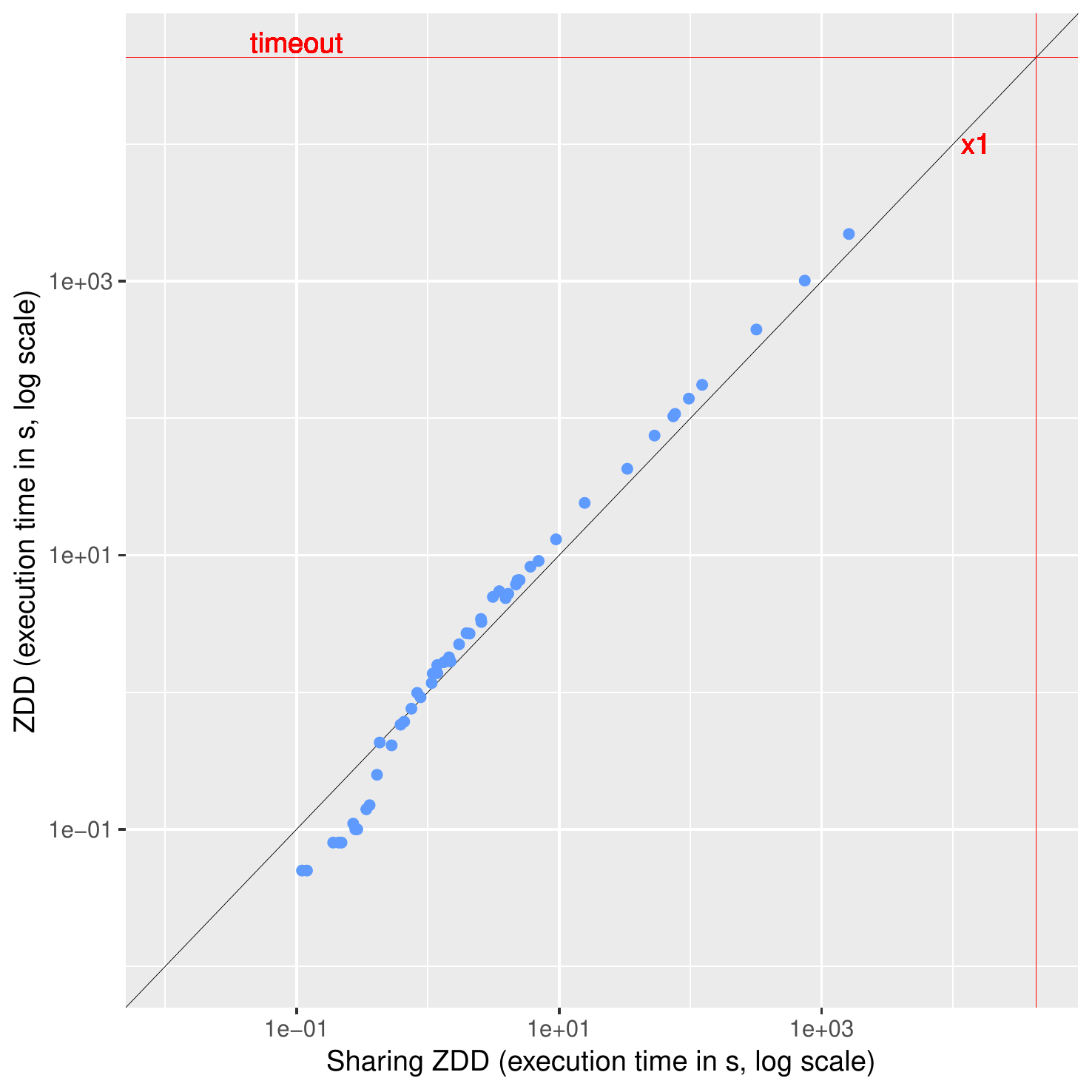}
		\caption{Time comparison}\label{fig:time_sharing}
	\end{subfigure}
	\begin{subfigure}[h]{0.45\textwidth}
		\centering
		\includegraphics[width=\textwidth]{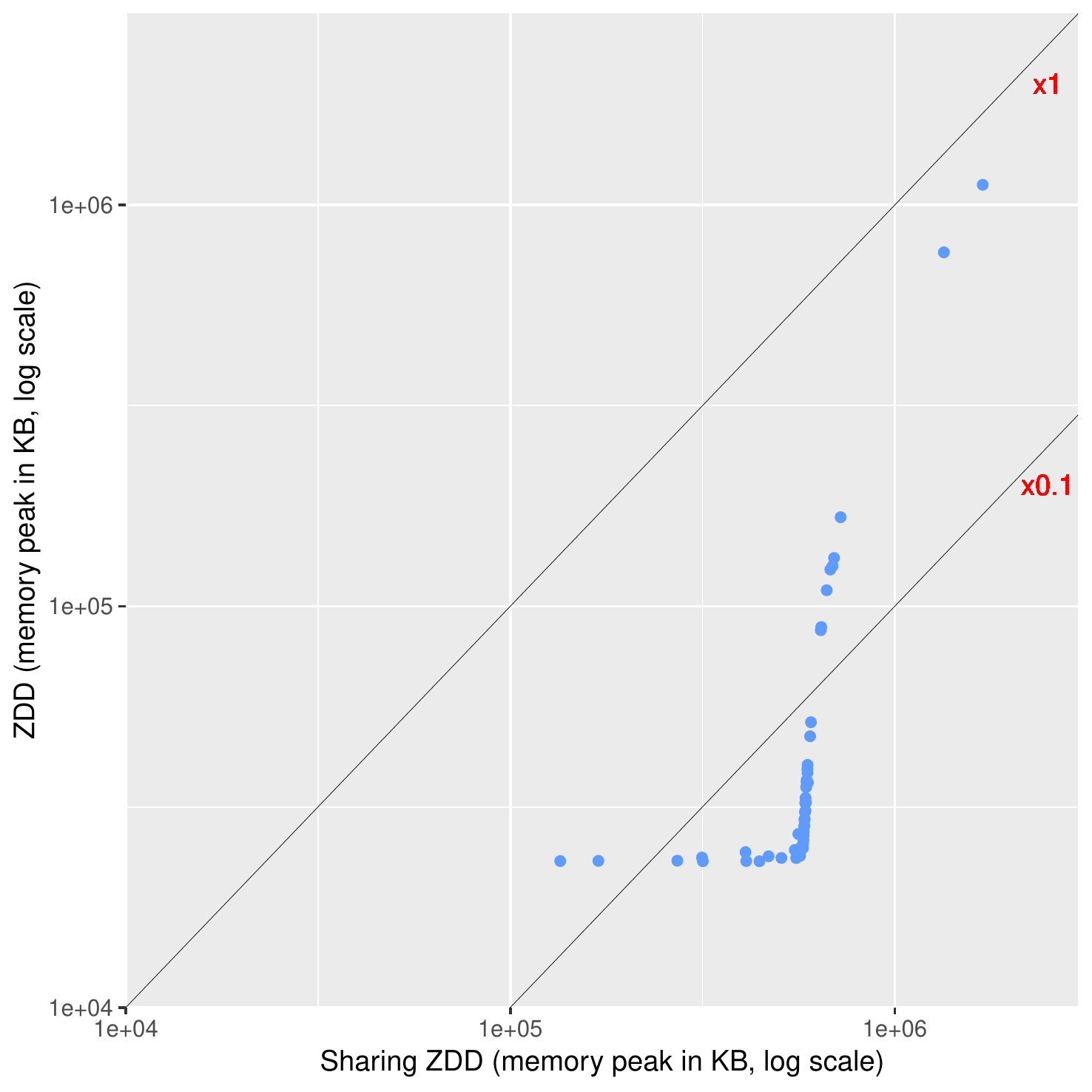}
		\caption{Memory comparison}\label{fig:memory_sharing}
	\end{subfigure}
	\caption{\emph{ZDD} vs. \emph{Sharing ZDD}}
	\label{fig:sharing}
\end{figure}

Comparing the \emph{ZDD} and \emph{Sharing ZDD} variants, we observe that sharing of ZDDs between analysis runs is not very beneficial in our case~(\Autoref{fig:sharing}):
\emph{Sharing ZDD} does not significantly improve the analysis execution time but is more costly in terms of peak memory usage.
It seems that there is not much sharing between successive analyses and that the library just fills the memory up to a certain threshold before collecting garbage;
our knowledge of \software{Cudd}'s internals is however insufficient to check whether this explanation is correct.

This motivates our yet unimplemented idea of parallelizing the analysis by running it for different values of $a$ on different cores, completely separate from each other.
Note that this is much easier than running analyses sharing a single ZDD manager, due to contention on the hash tables.

%% file: exact_analyses.tex
\section{Comparison of all exact analyses}

\begin{figure}[h]
\begin{center}
\includegraphics[width=\textwidth]{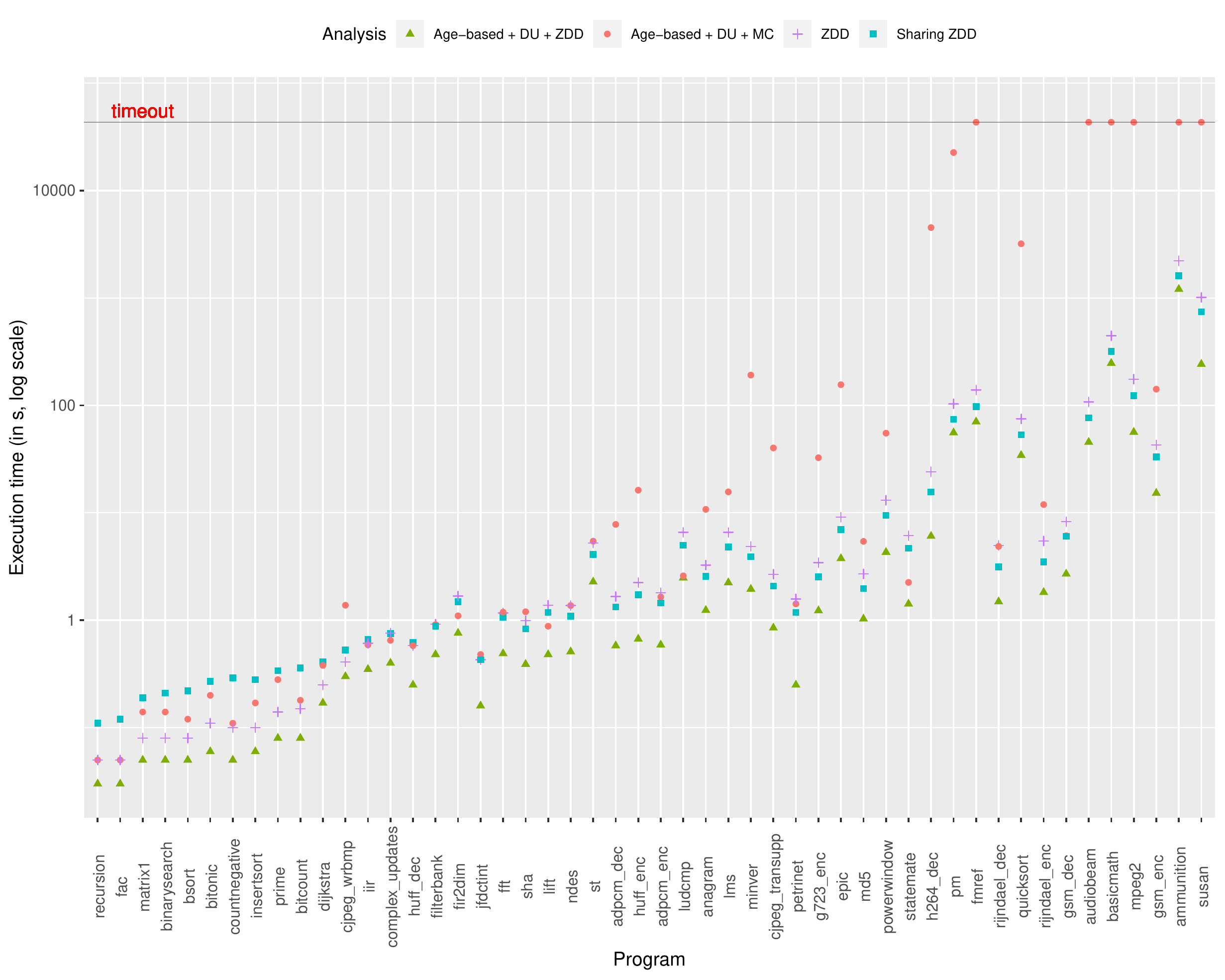}
\caption{Execution time of all exact analyses.}\label{fig:time_exact_analyses}
\end{center}
\end{figure}

\begin{figure}[h]
\begin{center}
\includegraphics[width=\textwidth]{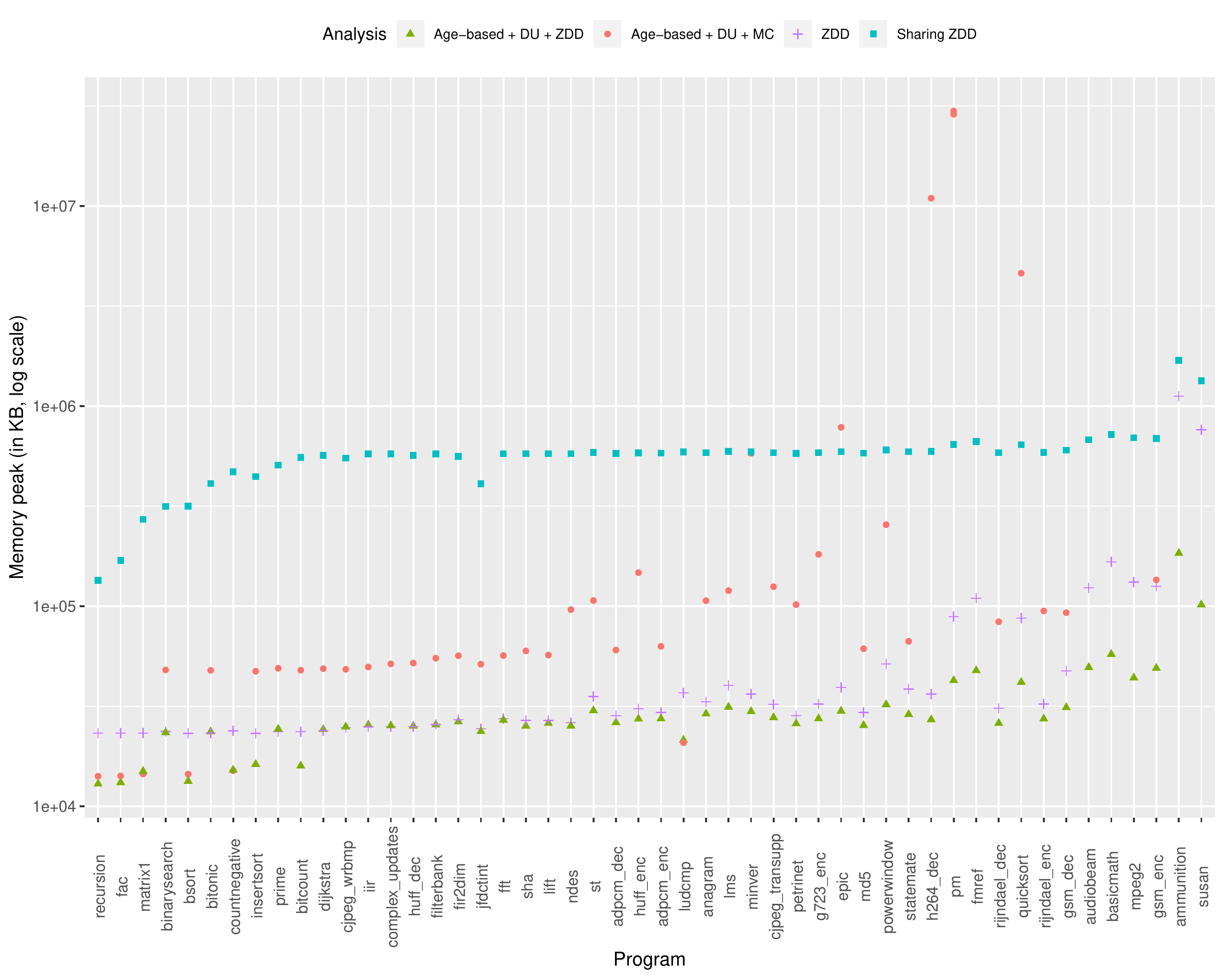}
\caption{Memory consumption of all exact analyses.}\label{fig:memory_exact_analyses}
\end{center}
\end{figure}

Figures~\Autoref{fig:time_exact_analyses} and~\Autoref{fig:memory_exact_analyses} show the execution time and memory consumption all fully-precise analyses we have implemented.